\documentclass{article}
\usepackage[utf8]{inputenc}
\usepackage[a4paper, total={6in, 8in}]{geometry}

\usepackage{amsmath,amsfonts,amssymb}
\usepackage{color}
\usepackage{graphicx}
\usepackage{cite}
\usepackage{float}
\usepackage{amsthm}
\usepackage{caption}
\usepackage{subcaption}
\usepackage{appendix}
\usepackage{comment}
\usepackage{authblk}

\theoremstyle{definition}
\newtheorem{definition}{Definition}[section]
\newtheorem{theorem}{Theorem}[section]
\newtheorem{lemma}{Lemma}[section]

\numberwithin{equation}{section}
\newenvironment{hproof}{%
  \proof}{\endproof}

\DeclareMathOperator{\Tr}{Tr}

\newcommand{\ket}[1]{| #1 \rangle}

\title{Hyper-Invariant MERA: Approximate Holographic Error Correction Codes with Power-Law Correlations}
\author[1]{ChunJun Cao\thanks{ccj991@gmail.com}}
\affil[1]{Joint Center for Quantum Information and Computer Science, University of Maryland, College Park, MD, 20742, USA}
\author[2]{Jason Pollack\thanks{jpollack@cs.utexas.edu}}
\affil[2]{Quantum Information Center and Department of Computer Science, University of Texas at Austin, TX, 78712, USA}
\author[3]{Yixu Wang\thanks{wangyixu@terpmail.umd.edu}}
\affil[3]{Department of Physics and Maryland Center for Fundamental Physics, University of Maryland, College Park, MD, 20742, USA}

\begin{document}

\maketitle
\begin{abstract}
    We consider a class of holographic tensor networks that are efficiently contractible variational ansatze, manifestly (approximate) quantum error correction codes, and can support power-law correlation functions. In the case when the network consists of a single type of tensor that also acts as an erasure correction code, we show that it cannot be both locally contractible and sustain power-law correlation functions. Motivated by this no-go theorem, and the desirability of local contractibility for an efficient variational ansatz, we provide guidelines for constructing networks consisting of multiple types of tensors that can support power-law correlation. We also provide an explicit construction of one such network, which approximates the holographic HaPPY pentagon code in the limit where variational parameters are taken to be small.
\end{abstract}

\tableofcontents

\section{Introduction}
Tensor networks are powerful tools to elucidate many-body physics \cite{Orus:2013kga,Orus:2018dya} and quantum gravity\cite{Swingle:2009bg,Chirco:2017vhs}. More recently, they have also been used to study quantum error correction both in the context of holography\cite{Pastawski:2015qua,Hayden:2016cfa,Cao:2020ksw} and beyond\cite{FerrisPoulin14,Harris2018,Farrelly:2020mxf,CaoLackey:2021}. For the former category, it was proposed that the Multiscale Entanglement Renormalization Ansatz (MERA) \cite{Vidal2008} resembles a discretized version of the AdS/CFT correspondence \cite{Swingle:2009bg}. At the same time, the AdS/CFT correspondence itself is believed to implement semi-classical bulk physics as a quantum error correction code (QECC) \cite{Almheiri:2014lwa}, where the bulk low-energy subspace defines a code subspace. A number of models of this correspondence have been given using tensor networks \cite{Pastawski:2015qua,Yang:2015uoa,PastawskiPreskill,Hayden:2016cfa,Bao:2018pvs,Kohler:2018kqk,Cao:2020ksw}. 
See \cite{Jahn:2021uqr} for a recent review.

While these constructions
(MERA\cite{Vidal2008}, HaPPY \cite{Pastawski:2015qua}, random tensor networks \cite{Hayden:2016cfa}, etc.) all capture some desirable features of the AdS/CFT correspondence, it is somewhat unclear how to unify these features under a single framework. MERA is an efficient variational ansatz capable of capturing the correlations and entanglement of the ground state while realizing an approximate quantum error correction code \cite{Kim:2016wby}.  However, it does not fully capture the correct symmetries \cite{Beny:2011vh,Czech:2015kbp,Bao:2017qmt} of AdS/CFT and is somewhat at odds with gravitational expectations for entropy bounds \cite{AdSMERA2015}. These flaws are amended in the hyper-invariant tensor network \cite{Evenbly2017}. However, in neither construction is there a clear physical picture of entanglement wedge reconstruction or code properties. For example, in MERA it is not clear how explicit logical operations can be realized on the boundary, or, equivalently, how to reconstruct some bulk operator from boundary operators. Although one can study cursory properties of the code\cite{Kim:2016wby}, its decoding and error correction properties are far from manifest. As for the hyper-invariant tensor network, it is not clear how it should be interpreted as a QECC.

The HaPPY code \cite{Pastawski:2015qua} on the other hand, being a stabilizer code, is a clear-cut QECC with well-defined logical and decoding operations\cite{Harris2018,Cao:2020ksw,CaoLackey:2021}. It also excels in reproducing the desirable features of subregion duality and bulk reconstruction. However, it is very limited as a variational ansatz for studying many-body quantum states. Without assistance from bulk correlations, it also produces trivial boundary two-point functions. Even with tuning of the bulk state, the boundary correlations cannot be made consistent with that of a CFT\cite{Gesteau:2020hoz}. Nevertheless, it is possible to produce some non-trivial correlations sparsely with localized operators\cite{Jahn2020}.
In the limit of large bond dimension, random tensor networks \cite{Hayden:2016cfa} can also produce a QECC and can sustain power-law decaying correlations. In principle, given a suitably chosen tessellation of the (hyperbolic) bulk, they also capture the symmetries of an underlying geometry. However, the randomness and the large bond-dimension limit inherent in their definition make it difficult for these networks to function as variational ansatze usable for many-body physics. Code properties and decoding also remain underexplored in this picture.

In this work, we bridge these gaps by creating what we call a hyper-invariant MERA (HMERA) tensor network that combines the features of MERA and those of a holographic QECC. The HMERA is a variational ansatz that can be efficiently contracted to allow computation of correlation functions in a way that is similar to the traditional MERA. As a(n approximate) QECC, the network is able to support the correct power-law behaviour for the two-point function, as well as features of subregion duality/bulk reconstruction. We provide a set of general construction guidelines for building these HMERA networks and provide one explicit example which is manifestly a quantum error correction code and well-approximates the HaPPY pentagon code\cite{Pastawski:2015qua} if we choose some of the variational parameters to be small. We also discuss some general constraints such codes follow. In particular, we prove a no-go theorem, which states that locally contractible HMERAs constructed from any single kind of quantum erasure correction code must contain trivial correlation functions. Hence our explicit construction, with two distinct types of tensors, is one of the simplest possible networks to exhibit power-law correlations.

We provide some basic definitions in Section \ref{sec:review}. In Section \ref{sec:general}, we discuss the constraints necessary for the HMERA to exhibit non-trivial correlation functions. We state a no-go theorem (proved in Appendix \ref{app:proof}) which requires that an HMERA with non-trivial correlators must necessarily contain more than one type of tensor. In Section \ref{sec:model}, we, therefore, turn to the construction of models with multiple types of tensors, giving both general guidelines for constructions and an explicit realization (with additional details postponed to Appendix \ref{app:details}). Finally, in Section \ref{sec:discussion} we finish with a discussion and concluding comments.
%\textcolor{red}{YW:shall we write a previewing section: the paper is organized as follows? Or incorporate this into the last section?}

\section{Definitions}\label{sec:review}
Typically, a MERA in the literature often refers to a multi-scale tensor network consisting of two types of tensors, the disentangler and isometry. However, the broader definition of MERA can extend to more general tensor networks such as the HaPPY code\cite{Pastawski:2015qua}, the hyper-invariant tensor network\cite{Evenbly2017}, and other similar network geometries. To avoid this potential ambiguity, here we define MERA networks that are consistent with some hyperbolic tiling \emph{hyper-invariant MERAs}, or HMERAs. We use MERA to specifically refer to the constructions described in \cite{Vidal2008}.

%\textcolor{red}{more stuff on mera?}

\begin{definition}
A \emph{hyper-invariant MERA (HMERA)} is a tensor network built on the dual graph on a uniform hyperbolic tessellation such that the tensor layouts are consistent with the discrete symmetries of the tessellation. 
%The tensor network is a pseudo-HMERA if distinct tensor values are used for different tiles such that one breaks the manifest symmetries of the tessellation.
\end{definition}
%\textcolor{red}{regular tilings are subsets of 1-uniform tilings.}
Examples of such HMERAs include the HaPPY code and hyper-invariant tensor networks. In these constructions, all tensors associated with the tiles, vertices and edges are chosen to be the same. However, one can, in general, place different tensors in various locations as long as they have the right degrees, i.e., the number of dangling edges; see e.g. Figure \ref{fig:HMERA model} below for an example. The usual MERA (see for example figure 7 of \cite{AdSMERA2015}), on the other hand, is \emph{not} an HMERA because the underlying tiling is not uniform. 

\begin{definition}\label{def:k-isometry}
A tensor of degree $N$ with constant bond dimension on each leg is \emph{$\ell$-isometric} if contracting $N-\ell$ legs of the tensors with its conjugate transpose yields $I^{\otimes\ell}$. In addition, it is \emph{permutation-invariant} if it is $\ell$-isometric for any such $(N-\ell)$-leg contractions.
\end{definition}

Graphically, this property is shown in Figure~\ref{fig:l_isometry}.
\begin{figure}[H]
    \centering
    \includegraphics[width=0.8\textwidth]{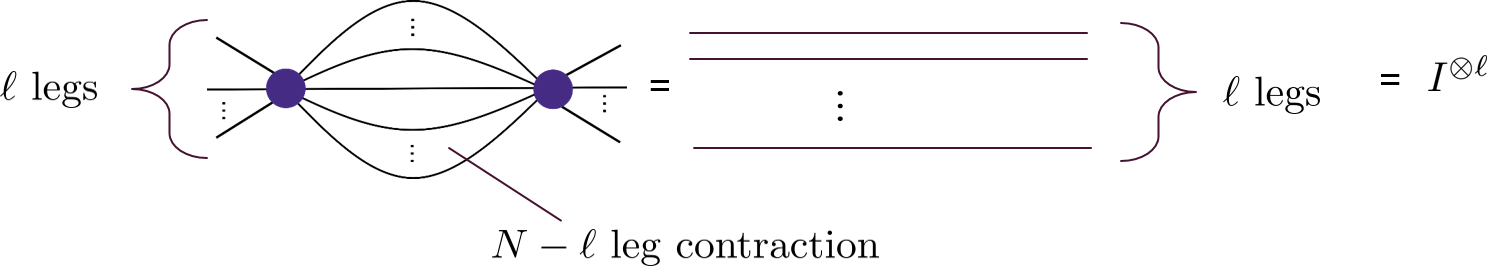}
    \caption{An example of a tensor where $N-\ell$ legs are contracted. Straight lines denote identity operators.}
    \label{fig:l_isometry}
\end{figure}

One can use such isometries as encoding isometries of quantum error correction codes. Here we call something a code when it encodes at least one logical degree of freedom (more generally, $k$ of them). This, for example, rules out the $k=0$ codes which simply encode a state. The encoded logical degree of freedom is often referred to as the \emph{bulk} degree of freedom in holographic tensor networks. We use these terms interchangeably in this work. 

Isometric properties can also be translated into error correction properties. If a degree $N$ tensor is permutation invariant $\ell-$isometric, then it can also serve as an encoding isometry for an $[[n,k,d]]$ code where $n=N-k$, $k\geq 1$, and $d\geq \ell-k+1$ \cite{CaoLackey:2021}. This is because when \emph{any} $\ell-k$ leg subsystem is maximally mixed, the map onto a subsystem is also a code that corrects at least that many erasures. Therefore any operation that alters the logical information must have support over at least $\ell-k+1$ sites. Here we use this notation to denote a code over $n$ qudits each with size given by the bond dimension of the tensor.
 If the tensor is not permutation invariant, then the resulting code can still correct such erasure on the \emph{specific} $\ell-k$ legs required by the isometry. However, in this case no conclusive statement can be made about code distances in general.
 
 Conversely, it is also straightforward to find examples of tensors with such isometric properties. For instance, any non-degenerate code with $d=\ell+1$ satisfies the property that the reduced density matrix on $\ell$ qubits is maximally mixed. This yields a permutation invariant $\ell$-isometry. The tensor that corresponds to this code is at least a $k+\ell$-isometry.

%Because any $l$-leg subsystem is maximally mixed, we can take it to be an isometry that encodes $k$ logical degrees of freedom into $N-k$ degrees of freedom. 
%However, one must be careful that only erasure of those particular legs are guaranteed to be correctable. By extension, if a tensor's isometric properties are invariant under permutations, then it is also an erasure correction code that corrects any $l$ erasure errors.

\begin{definition}
An HMERA is \emph{regular} if it is built from a regular hyperbolic tessellation. If the tensors over different tiles in this network are identical and permutation-invariant, then we say it is \emph{completely regular}.
\end{definition}

Note that a (1-)uniform tiling can be regular or semi-regular. For example, the semi-regular pentagon-hexagon code (Figure 17a of \cite{Pastawski:2015qua}) and the hybrid holographic Bacon-Shor code\cite{Cao:2020ksw} are completely semi-regular HMERAs. 
Here we focus on regular tilings where all polygonal tiles are identical. A completely regular tensor network will have individual tensors that obey the symmetries of the tiling. A completely regular HMERA is then a tensor network constructed by gluing together tensors that are also manifestly erasure correction codes. The HaPPY pentagon code and the holographic Steane code\cite{Harris2018} are such completely regular HMERAs.

\begin{definition}
An HMERA is \emph{locally contractible} if each tensor contracts to an isometry.
\end{definition}
There are cases in which groups of tensors together contract to an isometry while each tensor as its own fundamental unit does not; see, for example, \cite{Evenbly2017}. If the (groups of) tensors in the network do not contract to isometries, then it becomes more difficult to contract the network exactly when one computes quantities such as expectation values. This often leads to a much more costly algorithm, although there are instances where approximate contractions can also be performed in polynomial time\cite{PEPS,Levin2007}.

%need to show that this requires tensors to be 2-k isometries:
\begin{lemma}
A completely regular HMERA is locally contractible only if it has tensors that are least 2-isometric.
\label{lm:l2}
\end{lemma}
\begin{proof}
For a given layer of the hyperbolic tiling, there is at least one tile which has two edges facing inwards whereas all others have one edge facing inwards. As a result, the dual network contains loops. A tensor on this tile will have at least two legs facing inwards. If there are only 1-isometries, then it is impossible to contract this tensor from the outside, thus breaking local contractibility.
\end{proof}
Note that it is possible to generate completely regular tree tensor networks where all tensors are identical 1-isometries that are locally contractible. This lemma also ensures that a 5 qubit code is the simplest tensor that can be used as a building block for a locally contractible and completely regular HMERA that is also a non-trivial QECC. This is because the quantum Singleton bound requires that $n-k\geq 2(d-1)$. Complete regularity along with 2-isometry now ensure that each tensor block has $d\geq 3$. To encode a non-trivial amount of logical information, we also need $k\geq 1$. Hence this puts $n\geq 4+k\geq 5$. Hence the simplest building block for such a HMERA is indeed the $[[5,1,3]]$ code.

\section{General Constraints and no-go theorem}\label{sec:general}
%\textcolor{red}{do we still talk about other general constraints or other expectations of holographic codes, e.g. RT formula etc}

When considering the holographic code, a ``good'' non-degenerate QECC that corrects two or more erasure errors necessarily has a trivial connected 2-point function on the boundary. This is because in order for the erasures to be correctable, they must contain none of the encoded logical information. Therefore, for non-degenerate codes, the state $\rho_{AB}$ on the erased systems $A,B$ on the boundary must be maximally mixed. Such states have zero mutual information between the two erasures. Hence, for unit norm operators $O_A, O_B$ acting on $A, B$ respectively,

\begin{equation}
   \frac 1 2 (\langle O_A O_B\rangle-\langle O_A \rangle \langle O_B\rangle)^2 \leq I(A:B) = 0.
\end{equation}
This is, for example, the case for most two-site combinations in the HaPPY code. Therefore, to produce non-trivial correlation functions, one would have to consider a ``bad'' code with properties that are usually undesirable for quantum error correction. Additionally, it is difficult for it to properly capture the correct behaviour of subregion duality in holography.

This is not to say that it is impossible to produce a useful holographic code with non-trivial correlations. For example, degenerate codes may still sustain non-trivial two point functions while having the same code distance as a non-degenerate code. That is, let $\{|\bar{i}\rangle\}$ be an orthonormal basis for the code subspace; treating $O_A,O_B$ as two single-site errors, the Knill-Laflamme condition is still satisfied, 
\begin{equation}
    \langle \bar{i}|O_A O_B|\bar{j}\rangle = C_{AB}\delta_{ij},
\end{equation}
as long as $C_{AB}$ has off-diagonal elements while terms like $\langle\bar{i}|O_A|\bar{i}\rangle$ can vanish. Such is the case for the Shor code\cite{PhysRevA.52.R2493} but not for the $[[5,1,3]]$ code.

Another option is to consider approximate quantum error correction codes (AQECC), where the Knill-Laflamme condition is only satisfied approximately. This is a more natural approach because AdS/CFT should be described by an AQECC when the gravitational coupling is finite\cite{Cotler:2017erl,Hayden:2018khn}. In this case, it is possible to have a code that has ``bad'' erasure correction properties in general, but is nevertheless sufficiently close to a ``good'' code, such that it reproduces the desired behaviours like entanglement wedge reconstruction to leading order while also supporting correlation functions which become nontrivial as a result of $1/N$ corrections. 

It might seem straightforward to construct such (A)QECCs using the same techniques of \cite{Cao:2020ksw}, where one replaces a code by its ``noisy'' counterparts. However, it is more difficult to maintain both non-trivial correlation functions and exact contractibility at the same time. On the one hand, exact contractibility requires (Lemma~\ref{lm:l2}) that the tensor be at least 2-isometric, meaning that certain two-site subsystems need to be maximally mixed. Heuristically, this very property makes it harder for the tensor network to sustain non-trivial two-point functions. On the other hand, adding ``noisy'' terms to the tensor network tends to spoil said isometric properties, making it easier to support non-trivial correlations but harder to contract exactly. 
Therefore, the challenge to building an exactly contractible (A)QECC model with non-trivial correlations is in balancing these two opposing forces such that each tensor is just good enough an error correction code, or, more generally, isometry, to give exact contractibility, but not so good as to require trivial correlation functions on the boundary.

%In order to sustain the power-law correlation, the code must only correct these erasures approximately. Indeed as we see in the hyperbolic pentagon code, there is no power law decaying two-point function as they are trivial. Notably, it is possible to have power law decaying correlation functions on average where the operators are not necessarily weight two but are localized to two groups which produces a discrete version of the boundary CFT\cite{Jahn2020}. 
%We will see that the tension between these requirements can imply that 
The tension between contractibility and non-trivial correlation function in these holographic codes can be summarized in the following theorem. The simplest efficiently contractible construction which captures the discrete symmetries of the hyperbolic space is a locally contractible regular HMERA.

\begin{theorem}\label{nogo}
    A locally contractible and completely regular HMERA always contains a trivial connected boundary two-point function.
\end{theorem}

\begin{hproof}
By following the tensors from the interior out to the boundary, we can re-organize the network into different layers of tensors in a self-similar network which contains loops. Then one- and two-point functions can be computed by inserting an operator on the boundary and contracting the tensor network. The operator contracted with tensors in each layer is a coarse-graining operation as it moves up the layers, similar to conventional MERA ascending super-operators. However, because of the error correction properties of the tensors that make up the network, some of these operators will eventually be coarse-grained into zero or the identity operator. Hence the connected two-point correlation function vanishes.
\end{hproof}

In other words, if one constructs a tensor network out of a single type of code that corrects erasures exactly and demands it to be efficiently contractible, then some correlation functions of this network must be trivial. The details of the proof can be found in Appendix~\ref{app:proof}. 
%This is to say that the simplest efficiently contractible HMERA that produces a power-law decaying correlation function, regardless of whether it is exact or approximate QEC, cannot consist of only one type of tensor. 
Indeed, we see that this is true for both the hyper-invariant tensor network and MERA where at least two types of tensors are used in its construction. This is also true for the pentagon and the heptagon code. In fact, we see that there are trivial correlation functions in \cite{Jahn2020} even for localized two-point function and that the (pairs of) sites with non-trivial correlators are sparse in the infinite-layer limit\cite{Gesteau:2020hoz}. 

Note that here we have only focused on locally contractible completely regular HMERAs. It is in general possible to allow a completely regular HMERA to still be efficiently contractible by taking groups of tensors to contract to an isometry. We call such tensor networks quasi-locally contractible. 
One may worry that the local contractibility is too restrictive a requirement, and that there may be quasi-locally contractible networks that are completely regular. However, if a completely regular HMERA is quasi-locally contractible, then we can group the tensors into isometries. If the grouping produces another uniform tiling of the hyperbolic space, then it is simply a locally contractible uniform HMERA but with a different component tensor. By the above theorem it cannot have non-trivial correlations. If the network admits grouping of tensors into different isometries which correspond to a non-regular tessellation of the hyperbolic plane, then it is equivalent to an HMERA with multiple component tensors similar to MERA or the hyper-invariant tensor network, which are not completely regular.

The intuition of the above result is essentially that one has put too much ``isometric-ness'' into the network. By requiring the tensor network to be completely regular we force each tensor to be permutation invariant. This then forces every two-site subsystem on each tensor to be maximally mixed. As we have mentioned above, such properties prevent a non-trivial two-point function. It is clear that we then need to dial down the amount of ``isometric-ness'' in the network to avoid this no-go result. 

One can achieve this by removing different restricting clauses in the theorem while still having a satisfactory HMERA. For instance, this is possible by relaxing a completely regular network to a regular one. In doing so, we may relinquish some of the local symmetries of the tensor by not requiring permutation invariance. Such tensors can only be isometries when contracted in certain directions (e.g. the isometries in MERA). Alternatively, we can allow more than one type of tensor in the network while still demanding that all types have the same degree. For such kind of networks, one can selectively reduce the ``isometric-ness'' of some tensors such that we still maintain local contractibility. Of course, we also give up some symmetries in the network. Or, one can take a combination of these two approaches. The latter will be the approach we take in the next section. 

Going beyond regular and locally contractible HMERA, one can give up regularity entirely by going to semi-regular or $k$-uniform tessellations where different shaped tiles can map to different tensors. Local contractibility can also be relaxed to quasi-local contractibility, \textit{e.g.}, in the hyper-invariant tensor network, while still retaining an exactly contractible ansatz.

%Therefore, in order to have some nontrivial remaining correlation, we need to minimize the presence of such 2-isometries. Accordingly, we can allow 2-isometries only on certain legs. This way, we can preserve a lot more of the much needed non-trivial correlations. Of course, in doing so, we need to give up some symmetries, both globally, in how different tensors can be arranged in the network and locally in the tensors, such that they are only isometric in certain directions. Indeed, in the next section, we take some of these steps to minimize the ``isometric-ness''.

%There are several ways to circumvent the above no-go result. One solution is to give up local contractibility, but as a variational ansatz it is a much desired property. Another solution is that we can introduce more than one type of tensors. This can be done in different ways as well. One can have different tensors in a regular tessellation, which somewhat disconnects the symmetries of the tensor network from the symmetries of the tessellation. Relatedly, one can introduce less symmetric tensors that are not permutation-invariant and by sacrificing the exact erasure correction properties. Alternatively, one can consider a semi-regular or $k$-uniform tessellation where only different shaped tiles map to different tensors.

\section{HMERA model}\label{sec:model}
\subsection{General Construction Guidelines}
In this section, we consider a construction that uses more than one type of tensors while preserving a regular tessellation. 
Although it can be tricky to prepare a tensor network with desired properties using only one type of tensor, it is much easier with two or more. 

Let us begin with a regular tiling of the hyperbolic plane with Schlafli symbol $\{p,q\}$. We choose $p,q$ such that the suitable isometries which we will describe in more detail later exist. In addition, all tiles except the one at the center can be divided into two types: the polygons with one edge facing the center, or edge polygons (EP), and polygons with a vertex/two edges facing the center, or vertex polygons (VP)\footnote{This construction method will not apply to the cases where some polygons are neither EP nor VP, e.g. when $q$ is odd. However, for $q\geq 5$, one can circumvent this difficulty by preferentially designating such polygons to be either VPs or EPs. For $q=3$ the situation is much worse and it is unclear if such a simple modification is sufficient. }.

Correspondingly, we need two types of isometries, the $1(+k_1)$-isometries of degree $p+k_1$, which we assign to the EPs, and $2(+k_2)$-isometries of degree $p+k_2$, which are assigned to the VPs. For the top tensor living in the central tile, we assign a $k_0$-isometry which encodes $k_0$ qudits into $p$ qudits. From the error correction perspective, the top tensor does not exactly correct any erasure errors, but it may correct them approximately. By construction, we will assume that $p,q$ are suitably chosen such that these types of isometries exist (we will give an example of such a choice in the next paragraph). $k_0,k_1,k_2\geq 0$ are the number of ``bulk'' or logical degrees of freedom we want to assign to each tile. For the tensor network to be a non-trivial code, we want $k_0+k_1+k_2>0$. We then orient the remaining $p$ tensor legs such that each leg is perpendicular to the polygon edges. The bulk legs will be left uncontracted while the two in-plane legs of each tensors that lie on the same edge of a polygon will be contracted through the usual tensor contraction procedure\cite{Orus:2013kga,Pastawski:2015qua,Cao:2020ksw}. This then produces a tensor network that maps the bulk degrees of freedom onto the boundary degrees of freedom. For the sake of convenience, when we refer to isometries from now on we will automatically drop the $(+k)$ bulk degrees of freedom and only consider in-plane legs unless otherwise specified. 

For a slightly more concrete example, consider the $\{5,4\}$ tiling of the hyperbolic plane by pentagons in Figure~\ref{fig:HMERA model}, where the VPs are labelled by squares and EPs are labelled by disks. The top tensor, a (0(+1))-isometry of degree $6$, is denoted by a pentagon. Then for each VP, we can assign a (2(+1))-isometry and for each EP a (1(+1))-isometry, both of degree $6$. This will allocate $k_0=k_1=k_2=1$ bulk qubit degree of freedom for each pentagonal tile. Such isometries clearly exist. One example is to take a $[[5,1,3]]$ code for the VP and a $[[5,1,2]]$ code for the EP. For the top tensor, more generic encoding isometries $V:\mathbb{C}^2 \rightarrow (\mathbb{C}^2)^{\otimes 5}$ would suffice. For example, we can turn a non-additive $[[6,0,2]]$ code into a $[[5,1]]$ code which does not correct any erasure error by taking one of the tensor legs to be the bulk leg. Of course, while these isometries satisfy the necessary conditions we outlined, they need not be sufficient.

The error-correction properties of the code, e.g. the physical representations of logical operators, can also be easily derived using operator pushing. See \cite{Pastawski:2015qua,Cao:2020ksw} for details. The pushing rules follow from the isometric properties of the tensors, where an $\ell$-isometry of degree $p$ can push any operator supported on $\ell$ legs to the remaining $p-\ell$ legs. Therefore, the support of each logical operators can also be derived by following these local pushing rules in the tensor network. See Figure~\ref{fig:HMERA_push} for an example in the $\{5,4\}$ tiling where the relevant isometries are given by the $[[5,1,2]]$ and $[[5,1,3]]$ stabilizer codes\footnote{The operator pushing in Figure~\ref{fig:HMERA_push} also holds for the approximate code we construct in Section~\ref{explicitconstruction}. However, in Figure~\ref{fig:HMERA_push}b, only the logical Z operator can be supported on 4 legs of the imperfect code.}. 

\begin{figure*}[t!]
    \centering
    \begin{subfigure}[t]{0.5\textwidth}
        \centering
        \includegraphics[height=2.8in]{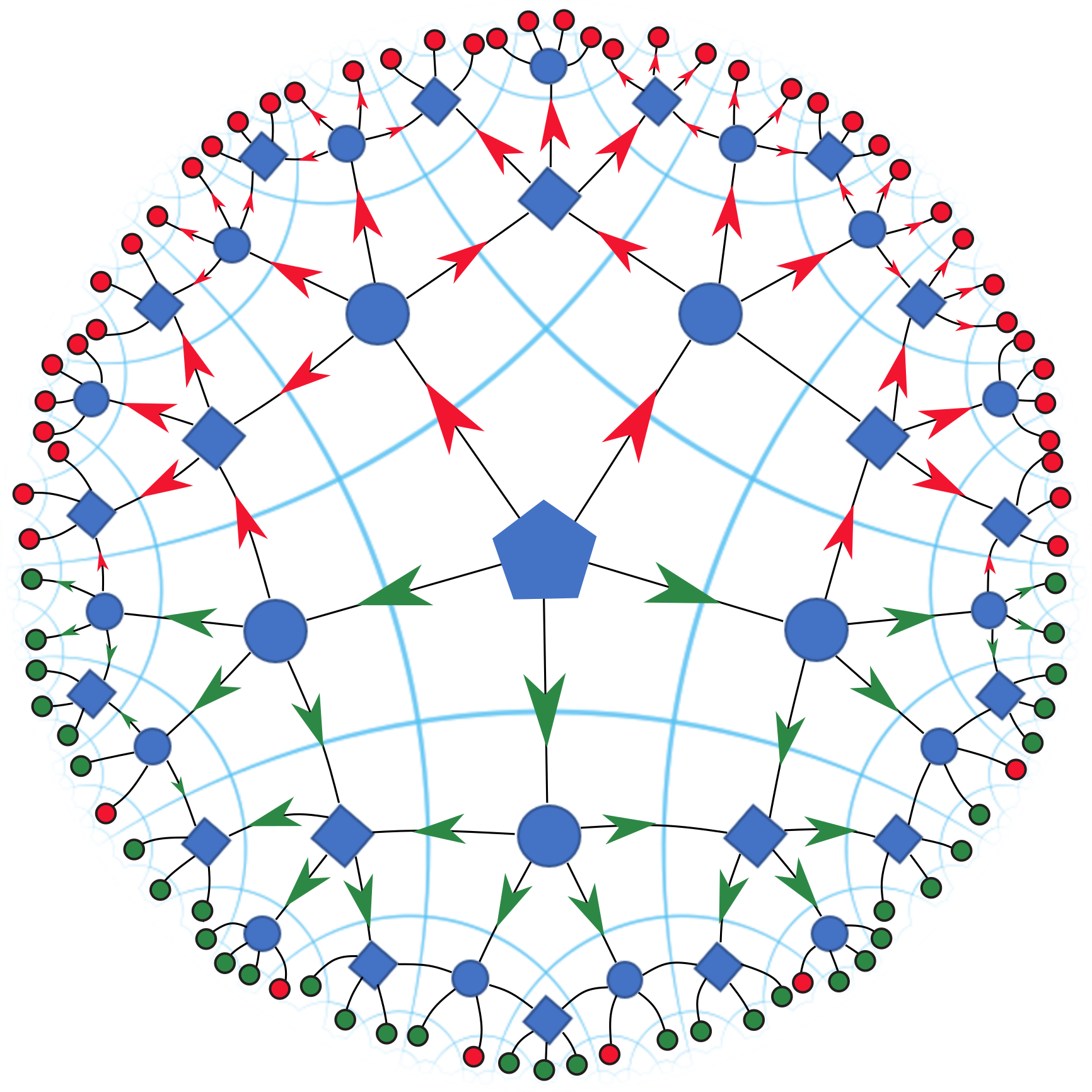}
        \caption{}
    \end{subfigure}%
    ~ 
    \begin{subfigure}[t]{0.5\textwidth}
        \centering
        \includegraphics[height=2.8in]{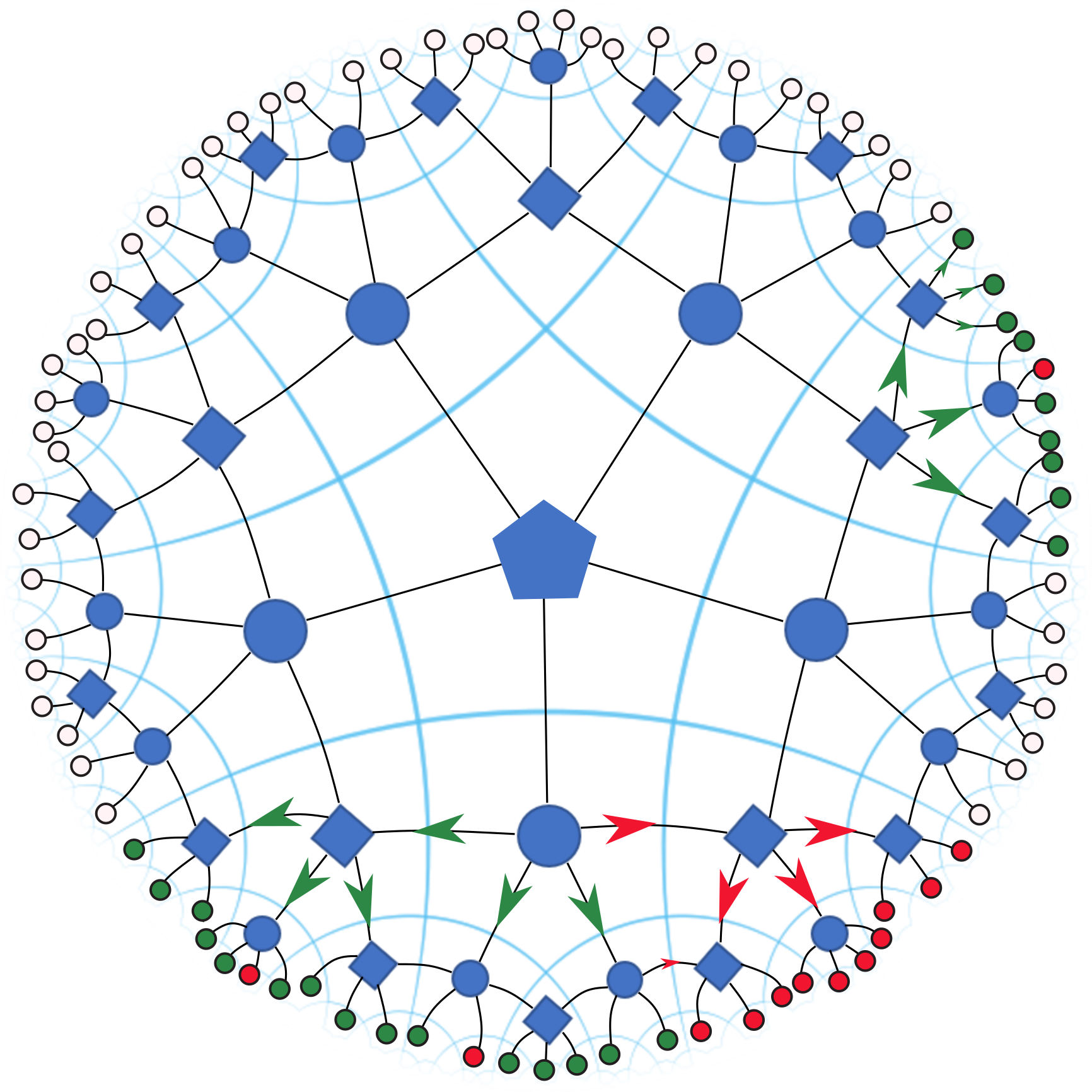}
        \caption{}
    \end{subfigure}
    \caption{(a) A logical operator on the central (top) tensor is supported on all boundary degrees of freedom (red and green nodes). However, if one constructs a code that approximates the HaPPY code, then a logical operator can be pushed along the green arrows to a subregion of the boundary (green nodes) with a small error. (b) A similar pushing for certain logical operators closer to the boundary produce operators that are supported on a subregion. The logical operator can be reconstructed approximately on green nodes and can be reconstructed exactly on the union of green and red nodes.}
    \label{fig:HMERA_push}
\end{figure*}

%\begin{figure}
%    \centering
%    \includegraphics[width=0.6\textwidth]{figures/HMERA_push.png}
 %   \caption{A logical operator on the central (top) tensor is supported on all boundary degrees of freedom (circles). However, if one constructs a code that approximates the HaPPY code, then a logical operator can be pushed along the green arrows to a subregion of the boundary (green circles) with a small error.}
 %   \label{fig:HMERA_push}
%\end{figure}
%\textcolor{red}{add figure for pushing to subregion}

Conversely, if all logical operators of some bulk region can be ``cleaned'' off of a boundary subregion, \textit{i.e.} there exists a representation of the logical operators such that they act trivially on said subregion, then the erasure of this subregion also does not affect the encoded logical information in that bulk region. Specific code properties will, of course, depend on the details of the isometries and the network one uses.

Furthermore, because these isometries are QECCs that correct no more than $e$ erasure errors, where $e=1$ for EPs and $e=2$ for VPs, they will not (exactly) correct any erasure errors on the boundary. Equivalently, this implies that an inserted weight one operator (such as a Pauli error) on the boundary will generally coarse-grain to a non-trivial (not zero or identity) operator on the more coarse-grained layers. This allows us to define non-trivial ascending/descending superoperators in a way similar to that of the MERA\cite{MERAalgo}. These are mappings that take operators to operators, corresponding to coarse/fine-graining a particular operator along the renormalization group direction. However, a key difference is that these superoperators can be different on different layers. Therefore we only take them to be similar to their MERA counterparts in an average sense. Suppose such a coarse-graining (ascending) superoperator averaged over the layers is given by $\bar{S}(\cdot)$ which admits spectral decomposition with eigenvalues $\bar{\lambda}_{\alpha}$ and corresponding eigen-operators $\bar{\psi}_{\alpha}$,
then, we can reuse the same argument as is used to derive the MERA power-law correlation functions. See Section 3 of\cite{MERAalgo} for example. 

Consider inserting operators $\bar{\psi}_{\alpha},\bar{\psi}_{\beta}$ on the boundary sites indexed by $i,j$. Then the correlation function is expected to behave as

\begin{equation}\label{eqn:correlationinTN}
    \langle \bar{\psi}_{\alpha}(i) \bar{\psi}_{\beta}(j)\rangle \approx C_{\alpha\beta}\bar{\lambda}^l_{\alpha}\bar{\lambda}_{\beta}^l= \frac{C_{\alpha\beta}}{|i-j|^{\Delta_{\alpha}+\Delta_{\beta}}}.
\end{equation}
Here $C_{\alpha\beta}=\Tr[\bar{\Psi}(\alpha,\beta)\bar{\rho}]$ is the expectation value of a localized coarse-grained operator $\bar{\Psi}(\alpha,\beta)$ evaluated against the reduced state $\bar{\rho}$ supported on a few sites which is approximately a fixed point\footnote{The existence of a fixed point is an additional assumption we make for the model. It seems physically reasonable that, for such a tensor network with a large number of layers with the same average ascending/descending superoperators, there should be a fixed point.} of the coarse-graining $\bar{S}(\bar{\rho})\approx\bar{\rho}$, and under its dual fine-graining descending superoperator $\bar{S}^*(\bar{\rho})\approx\bar{\rho}$. Equivalently, one can keep coarse-graining the operators $\bar{\Psi}(\alpha,\beta)$ using the ascending superoperators all the way up to the top tensor then evaluate $C_{\alpha\beta}$ against the top tensor\footnote{The top tensor may encode global information of the state\cite{MERAtopological}, but its impact on a few-site reduced state is washed out by the ascending/descending superoperators, so we expect that (\ref{eqn:correlationinTN}) should not depend on the choice of top tensor. }. $l\sim \log|i-j|$ is the number of layers of coarse-graining needed before $\bar{\psi}_{\alpha}\bar{\psi}_{\beta}$ becomes localized, and  $\Delta_{\alpha}=\log\bar{\lambda}_{\alpha}$. The specifics of these super-operators will depend on both the network structure and the tensor isometries we use. 

Note that this heuristic argument, by itself, does not imply the model can be used to approximate the ground states of CFTs whereby $\Delta$ are fit to the correct primary operators. Our statement is simply that such networks should support power-law correlations, as opposed to other completely regular HMERAs which Theorem \ref{nogo} ensures do not.

Thus far, the method described constructs a quantum code which, in general, does not correct for erasures of subregions. In addition, the logical operators are not necessarily represented transversally. While such properties are not desirable for fault-tolerance, the latter is not a concern for models of AdS/CFT\footnote{In fact, it would be extremely unusual if all bulk operators can be represented transversally on the boundary, i.e. that they are simple tensor products of boundary operators acting on disjoint subregions, because the smearing function is highly non-trivial.}.
As for the former, holographic codes in general should correct for erasures of boundary subregions, at least approximately, in the large $N$ limit. Therefore, if one wants more similarity with holography, it is not enough for such tensor networks to simply be a ``bad code'', rather it also has to be ``close'' to a good holographic code with proper erasure correction properties like \cite{Pastawski:2015qua}. To this end, we have to take a bit more care in constructing the isometric tensors such that they are also close to good erasure correction codes with larger code distances. 

\subsection{An Explicit Construction}\label{explicitconstruction}
We now present a concrete example of HMERA that is both efficiently contractible and permits non-trivial boundary correlations on all sites. We will construct isometries using a tunable parameter $\theta$ such that the code reduces to two copies of the $[[5,1,3]$ ``perfect'' code when $\theta=0$. It well approximates the HaPPY code when $\theta$ small, but can now produce power-law decaying correlations and can in principle sustain a non-flat entanglement spectrum. We base our model on a tensor product of two copies of the HaPPY pentagon code. To circumvent the no-go theorem \ref{nogo}, we substitute some of the perfect codes in the network with imperfect codes which we now construct.

First we construct a $1(+0)$-isometry tensor associated with the edge polygons which we call the imperfect code.
The imperfect code is illustrated graphically in Figure \ref{fig:imperfect tensor illustration}. This is but one way of constructing such isometries which are manifestly close to a perfect code for some parameters. In practice, one can easily construct other tensors over larger bond dimensions with more variational parameters. However, we choose this simple construction for the sake of clarity. 

The imperfect code is a superposition of the double copy perfect tensor together with weight-2 Z-type errors inserted\footnote{The added terms with Z errors help tilt the perfect code away from a stabilizer code. In doing so, it adds magic, which is necessary at all scales of the tensor network to produce a low energy state of a CFT\cite{White:2020zoz}.}. A weight-1 Z error is inserted in each copy, but no two errors are inserted on the same leg of different copies. Here $i,j=1,2,...,5$ are labelling legs in which errors are inserted. All different possible insertion configurations are summed up. To have this imperfect tensor be a 1-isometry (see  Definition \ref{def:k-isometry}), the parameters should satisfy the following relation:
\begin{equation}
    \cos\theta^2+\sum_{i\neq j} \sin\theta_{ij}^2=1.
\end{equation}
To retain the original symmetry of the perfect tensor, we usually choose the parameterization $\sin\theta_{ij}=\frac{\sin\theta}{\sqrt{20}}$. Both the perfect and the imperfect tensor being isometries guarantees the efficiently contractible property of our HMERA model.

Treating each leg as a qudit, we can show that this imperfect tensor defines a $[[5,0,2]]_4$ code but approximates two copies of the perfect code (Appendix \ref{app:imptensor}). It exactly reduces to two copies of the 5 qubit perfect code when $\theta=0$, and therefore inherits all its error correction properties approximately when $\theta$ is small. In the tensor network notation, each leg has bond dimension $4$ and the code sub-algebra is supported on any three legs of this tensor when $\theta=0$. However, when $\theta$ small but nonzero, it only approximates the double-copy perfect code. 

\begin{figure}
    \centering
    \includegraphics[width=0.75\textwidth]{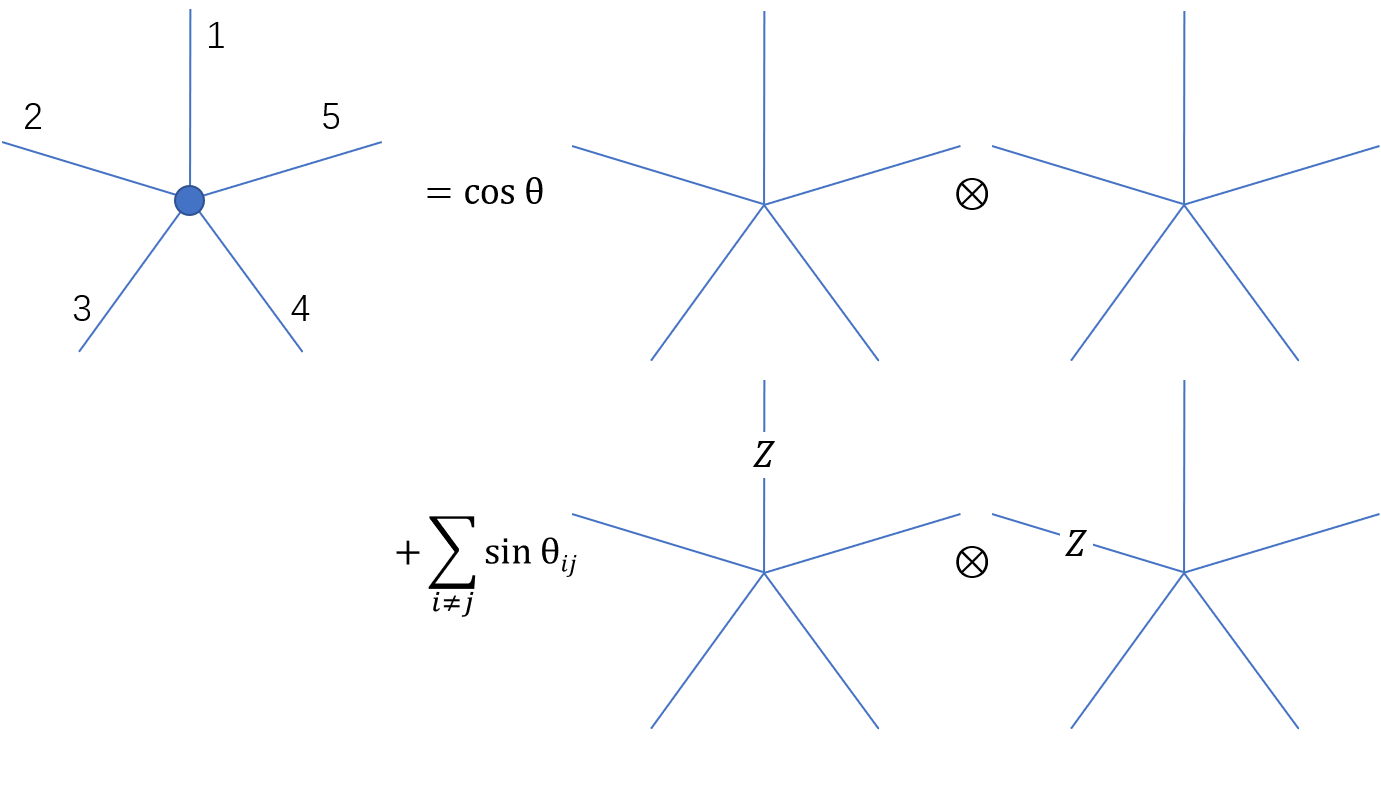}
    \caption{The construction of the imperfect code where each line on the left hand side represents two physical qubits. Each 5-pronged diagram on the right hand side is a perfect code. In the second line $i,j$ label the legs in which Z errors are inserted on each copy respectively.}
    \label{fig:imperfect tensor illustration}
\end{figure}

We can also construct the top tensor in a similar fashion by super-imposing the same codes with different choices of the code subspace. For instance, let $V$ be the perfect code encoding isometry, then we may construct a tensor that approximates the perfect code for small $\phi_{i\ne 0}$ with encoding map $V_T$ such that,

\begin{equation}
    V_T = \cos\phi_0 V +\sum_{i=1}^5 \sin \phi_i P_i V
    \label{eqn:toptensor}
\end{equation}
where $\cos^2 \phi_0 +\sum_{i}\sin^2 \phi_i =1$ and $P_i \in \{X,Y,Z\}$ are the same Pauli operators acting on site $i$. For the sake of symmetry, let us choose $\phi_i=\phi_0=\phi$. We can then take two tensor copies of this code as the top tensor we use in the network. 

To construct the full tensor network, we specify the positions of the substitution and replace those perfect codes with their imperfect counterparts. This is explicitly shown in Figure \ref{fig:HMERA model}. First we choose a pentagon as the origin and replace it by an imperfect double copy top tensor~(\ref{eqn:toptensor}). From the origin we can label each pentagon uniquely by the number of edges we need to cross to reach it from the origin. The collection of pentagons that are assigned to the same number a is called a ``layer''.  Apart from the top tensor, if a particular pentagon tile has only one edge connected to the previous layer, we replace the two tensor copies of the perfect code on it by the imperfect code~(Figure~\ref{fig:imperfect tensor illustration}). 

\begin{figure}
    \centering
    \includegraphics[width=0.5\textwidth]{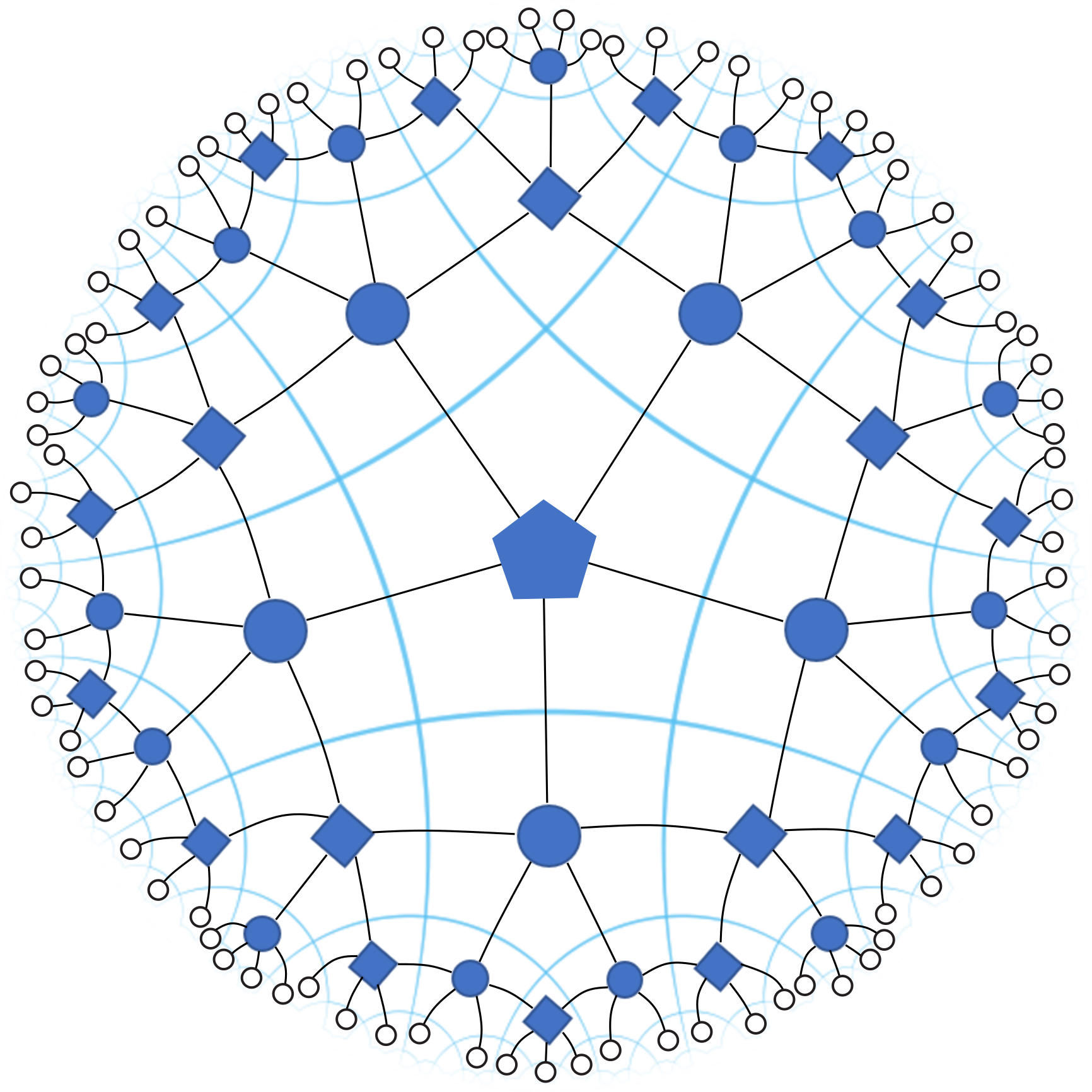}
    \caption{The illustration of the HMERA tensor network. The first several layers of substitution is explicitly shown. The circles represent the imperfect tensors defined in Figure \ref{fig:imperfect tensor illustration}, while the squares represent the tensor product of two perfect tensors.}
    \label{fig:HMERA model}
\end{figure}
%mention that while we could have chosen an additive distance 2 code, we chose a non-additive version so that there is some free parameter that we control, which allows us to recover HaPPY code in the parameter->0 limit. This is important because it gives us a tunable testing ground for understanding how A in the AQEC impacts features in holography, as it was shown in the Bacon-Shor paper that some version of the AQEC, such parameters play the role of G_N

%The model we construct above produces the power law decaying behaviour for the connected two point correlation functions. 
%For example, consider the two point function with operators inserted in the boundary legs indexed by $i,j$, The two-point function is expected to behave as

%\begin{equation}\label{eqn:correlationinTN}
%    \langle O_i O_j\rangle-\langle O_i\rangle\langle O_j\rangle \approx C_{OO}\tilde{\lambda}^{2N}\approx \frac{C_{00}}{|i-j|^{2\Delta}}.
%\end{equation}
%Here $C_{OO}\ne 0$ is some operator-dependent constant evaluated at the top tensor. $N\sim \log|i-j|$ is the number of the layers of coarse-graining. We can think of $\bar{\lambda}$ as the average of the dominant eigenvalue of the coarse-graining super operators over the layers:
%\begin{equation}\label{eqn:averageeigenvalue}
%    \bar{\lambda} = \prod_a \lambda_a^{p_a}.
%\end{equation}
%The index $a$ labels different types of super operators and $p_a$ denotes the probability of type $a$ super operator contributing to the coarse-graining process. 

The variational parameters of this model are given by the logical degrees of freedom, the top tensor skewing parameter $\phi$, and $\theta_{\tau}$ where in principle each imperfect code is labelled by $\tau$. For a more symmetric construction, we can set $\theta_{\tau}=\theta, \forall \tau$.

This network again supports a power-law decaying correlation function. Heuristically, it should contain at least one type of scaling operator with $\Delta \approx \log \bar{\lambda}$, where $\bar{\lambda}$ as the average of the dominant eigenvalue of the coarse-graining super operators over the layers:
\begin{equation}\label{eqn:averageeigenvalue}
    \bar{\lambda} \approx \prod_a \lambda_a^{p_a}.
\end{equation}
The index $a$ labels different types of super operators and $p_a$ denotes the probability of type $a$ super operator contributing to the coarse-graining process. It can be shown that one can choose the variational parameters such that $\lambda_a<1$. The detail distributions of $p_a$ and the details on the different superoperators can be found in Appendix~\ref{app:details}.

This HMERA model effectively ``interpolates'' between the usual MERA (or hyper-invariant tensor network model) and the HaPPY code model of holographic quantum error correction codes. By tuning these variational parameters, it mimics the former by introducing some variational parameters, power-law correlation functions while reproduces the approximate code properties of latter in reconstructing bulk operators on the boundary. Furthermore, because it is an approximate holographic quantum error correction code for small $\theta,\phi$, it also approximately retains all properties of the HaPPY code such as the Ryu-Takayanagi formula, entanglement wedge reconstruction, and other error correction properties.

\section{Discussion}\label{sec:discussion}
In this work, we have examined some general properties needed to construct a holographic quantum error correction code that can support power-law decaying correlation functions. By also demanding the network to be locally contractible, we found that it is impossible to satisfy both these requirements simultaneously in a completely regular HMERA. 
Instead, one has to introduce more than one type of tensors in order to construct a satisfactory variational ansatz that is also exactly contractible. This statement also coincide with the general observation that so far all the MERA-like variational ansatze presented contain at least two types of tensors. 

We have also provided general guidelines for constructing approximate holographic quantum error correction codes with the aforementioned properties. In particular, we gave one explicit construction where the tensor network approximated two copies of the HaPPY pentagon code in certain regimes of the parameter space. 

%summarize findings. Comment on entanglement spectrum, power and flexibility as variational ansatz, generalizability and future directions.
Although our proposed HMERA network can in principle serve as a variational ansatz, much work is needed to establish its utility. In particular, one would need to develop an optimization algorithm and analyze its complexity similar to \cite{MERAalgo}. We will leave this for future work. 

It is also desirable, both from AdS/CFT and from many-body physics, for such an ansatz to capture the CFT entanglement spectrum. We have not investigated this in detail, as the aforementioned algorithm and the associated numerics are not yet available. However, we can provide some general speculations using heuristic arguments. As shown in Appendix~\ref{subsec:entspec}, cuts through any edge of a perfect tensor have a tendency to flatten the entanglement spectrum while cuts through more than one edge of the imperfect tensors be the opposite. As the entanglement of a subregion is tied to minimal cuts through the tensor network, for large enough regions, we expect such a cut to contain both types of edges after averaging over the different types of tensors that can be contained in the subregion's past domain of dependence. Therefore, the overall entanglement spectrum should be somewhere in between the behaviour shown in Figure~\ref{fig:density matrix 3 type 1}
 and Figure~\ref{fig:density matrix 3 type 2}. Thus the tensor network should be able to accommodate non-flat spectrums.
 
We are optimistic that such methods can be generalized to create other HMERA-like variational ansatze which are also QECCs. In particular, because AdS/CFT is tied to AQECCs, we hope that such HMERA models can provide computable examples that help us understand some aspects of holography beyond leading order. They may also serve as useful tools for intuition building for small $N$, where the tensor network remains valid and well-defined despite being in the regime where gravity is strongly coupled and highly quantum.

\section*{Acknowledgements}
We thank Brian Swingle for helpful discussions and comments. C.C. acknowledges the support by the U.S. Department of Defense and NIST through the Hartree Postdoctoral Fellowship at QuICS, by the Simons Foundation as part of the It From Qubit Collaboration, and by the DOE Office of Science, Office of High Energy Physics, through the grant DE-SC0019380. J.P.
is supported by the Simons Foundation through \emph{It from Qubit:
Simons Collaboration on Quantum Fields, Gravity, and Information}. Y.W. is supported in part by the U.S. Department of Energy, Office of Science, Office of Advanced Scientific Computing Research, Accelerated Research in Quantum Computing (FAR-QC) and by the Air Force Office of Scientific Research under award number FA9550-19-1-0360.
\begin{appendix}

\section{Proof of the No-go theorem}\label{app:proof}
\begin{proof}
Consider a regular tiling of the Poincare disk with Schlafli symbol $\{p,q\}$, which tiles the plane with polygons of $p$ edges with each vertex of the polygon adjacent to $q$ $p$-gons. 

Let us tile the space as follows. We start by having the central tile (top layer) then gradually add more tiles layer by layer by radiating outwards. The first layer consists of the $p$ polygons whose edges are immediately adjacent to the central tile. Of the polygon in a layer, those that have edges facing toward the center we call \textit{edge polygons} (EPs). We then repeat the process by adding more EPs for the polygons in the outer layers. Finally, we finish by adding polygons that have a vertex facing the center which fills the gap between EPs, which we call the \textit{vertex polygons} (VPs). They have two edges facing inwards. If the tiles in an outer layer are immediately adjacent to a tile in the inner layer, then outer tiles are the``descendants'' of the inner tile, which is the ``ancestor.'' 

To construct a completely regular HMERA, for each tile we now place the same tensor with $p$ in-plane legs on the centroid of each polygon such that its legs are perpendicular to the $p$ polygon edges. For two polygons that share a common edge, the two tensor legs that cross this edge are contracted. By lemma~\ref{lm:l2}, the tensors are also at least 2-isometric for the network to be locally contractible.

Case 1 ($q>3$) :
For $q>3$, no two adjacent polygons in the same layer can have the same ancestor. Thus for $q>3$, all tiles can be divided into the above two categories\footnote{For $q>3$ but odd, there are edges of VPs that are neither inward nor outward facing. However, we can take such legs to be outward facing choosing some sequence of adding tiles. For instance, when two adjacent tensors are connected on the same layer (two nearby polygons share a common edge that is neither inward nor outward facing), we can take one of them to be EP which has one inward facing edge and the other to be VP which has two.}.

Each EP has one inward facing leg that connects to the direct ancestor, and each VP has two. Note that a VP cannot have direct ancestors that are only VPs. This is by construction, where we always add VP after growing EPs. Therefore, for a VP, at least one of its direct ancestor/descendent is an EP. 

 Note that we will only focus on operator insertions where the operator is not proportional to the identity. This is because the connected two-point function for identity operators always vanishes. Then, for a single operator insertion on the boundary, it is either inserted on a tensor on EP or a VP. 
 If the former, then the coarse-graining ascending super-operator can be reduced to contracting the 1-isometry with one operator insertion (Figure~\ref{fig:epvp_contract}a).
 Because the tensor corrects at least two located errors $d\geq 3$, this contraction is zero. One can see this by first tracing the other $p-2$ legs without operator insertions, which reduces to the identity map tensoring a single operator insertion of the form $\Tr[O\rho]$ where $\rho\propto I$. We can then decompose the operator $O$ into a part that is proportional to the identity and a part that is traceless. The first part coarse-grains to an operator proportional to the identity; they do not contribute to non-trivial correlations. The tensor contraction of the second part vanishes by tracelessness.  Because the network always has EPs in its outer layer, this is sufficient to show that the state generated by the network contains trivial correlation functions.
 
\begin{figure*}[t!]
    \centering
    \begin{subfigure}[t]{0.25\textwidth}
        \centering
        \includegraphics[height=1.5in]{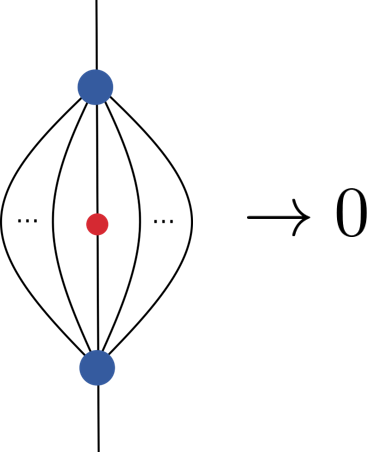}
        \caption{}
    \end{subfigure}%
    ~ 
    \begin{subfigure}[t]{0.75\textwidth}
        \centering
        \includegraphics[height=1.5in]{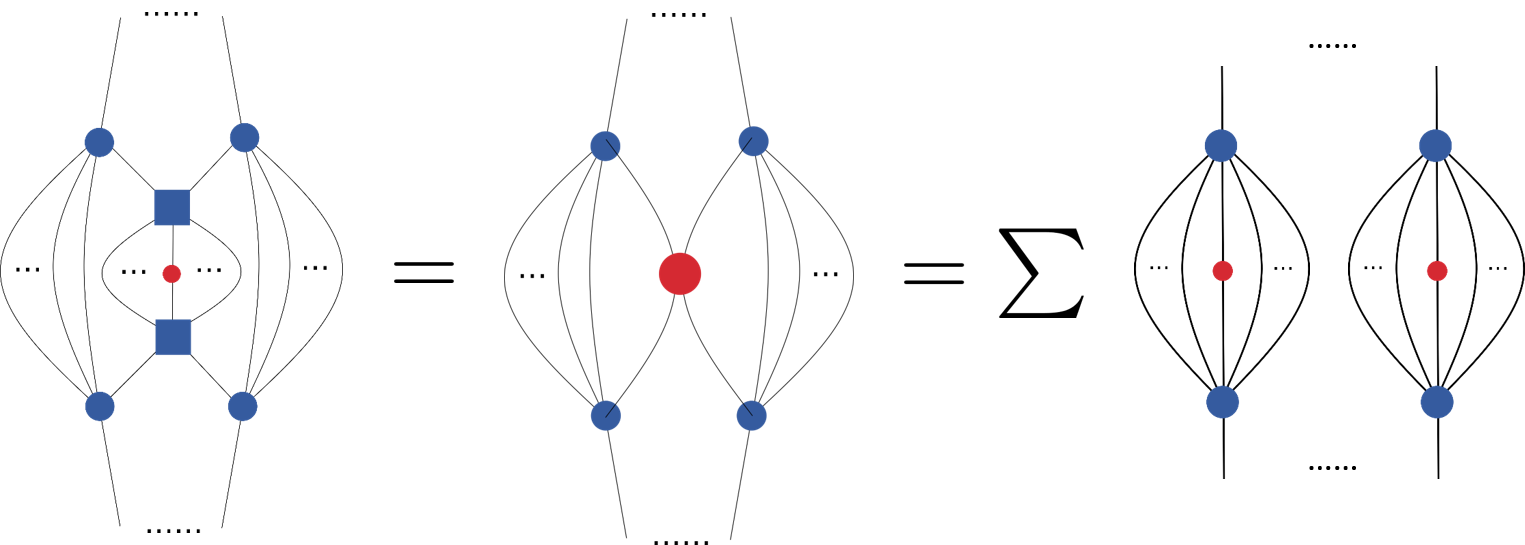}
        \caption{}
    \end{subfigure}
    \caption{(a) For inserted operators not proportional to the identity (red dot), such components are found on EPs and are contracted to 0. (b) An operator inserted on a VP tensor can be coarse-grained to a weight 2 operator insertion. It can then be decomposed as sum over weight 1 operators contracted on EPs.}
    \label{fig:epvp_contract}
\end{figure*}

For completeness, let us also examine operator inserted on VPs.
If the boundary operator is inserted on a VP, then either the tensor contracts to 0, when $d>3$, or it produces a weight 2 operator on the parent layer. The former again is trivial. For the latter case, because the graph is simple, the 2 operators must be passed onto 2 different tensors in the parent layer. At least one of these tensors is on an EP because they are immediate ancestors of a VP. They have a weight 1 insertion on each tensor, which makes the EP tensor contraction vanish also (Figure \ref{fig:epvp_contract}b). The only non-trivial term from such contractions is if the parent contains only one EP and that an identity operator is pushed to the parent EP while a weight one operator is pushed to the parent VP. However, this cannot always be true because the contraction terminates on a ``top tensor'' whose descendents are EPs only. Therefore, there must exist a layer for which both parents are EPs. Hence for $q>3$, we do not have any non-trivial super ascending operator. Therefore a two-point function is trivial as long as the two insertions are sufficiently far apart on the boundary.

%\textcolor{red}{add figures here for the q=3 tessellation?}
Case 2 ($q=3$):
By construction, each descendent of the central tile is immediately adjacent to two others in the same layer. Therefore, for a tensor network on this tiling to be contractible at all, it has to be at least a 3-isometry. Again, we drop the bulk indices for convenience. Adding the completely regularity condition, this implies $d\geq 4$. Thus for any operator inserted on a VP, which has only 2 edges facing inwards, insertion on these sites will vanish. Because there are VPs in this tiling, there will be two-point functions that are trivial. 

\end{proof}
\section{Details of the HMERA model}\label{app:details}
\subsection{Super-operators}
Recall that tensor networks of this type can yield a power-law decaying correlation as shown in (\ref{eqn:correlationinTN}) where the scaling dimensions are given by (\ref{eqn:averageeigenvalue}). To preview the results of this subsection, the different types of superoperators are depicted in Figure \ref{fig:superoperators} and their corresponding probabilities are presented in Eq. \ref{eqn:probability super operator}. As long as we show that $(1)$ the only eigen-operator for eigenvalue 1 is the identity operator; $(2)$ for each type of superoperator all the other eigenvalues indeed have absolute value less than $1$, we can conclude it produces a decaying connected two-point function that is roughly a power law.

We start by investigating the properties of each tensor individually, as the superoperators interpolating between layers are composed of them. The perfect tensor defines a $2$-isometry depending on the state of the logical qubit. See the left panel of Figure \ref{fig:perfect tensor}. As the two copies form a direct product, studying a single copy suffices. Most generally, we initialize the logical qubit in the state $|\phi\rangle=\cos\alpha |0\rangle+e^{i\beta}\sin\alpha\ket{1}$ and denote the corresponding isometry as $W_p$.  By distinguishing the inward 2 legs from the others, the original rotational symmetry breaks down to the reflection symmetry $1\leftrightarrow2,~ 3\leftrightarrow5$. Except for the top few layers, the incoming operator of $W_p$ will at most be weight 2, and will have support either on $3,4$ or $4,5$. This motivates us to view $W_p (\cdot) W_p^\dagger$ as a super-operator sending a weight 2 operator to another weight 2 operator. Studying the case in which operators have support on legs 3 and 4 would suffice because of  the reflection symmetry. If the operator is weight 1, say $X_3$, we view it as a weight 2 operator $X_3I_4$, etc. The identity operators together with Pauli matrices on both legs form a set of bases of weight 2 operators. The super-operator can be viewed as a $16\times16$ matrix in such bases. The right panel of Figure \ref{fig:perfect tensor} illustrates how to determine one entry of the super operator. One can easily calculate the eigen-operators and the eigenvalues of this superoperator. It turns out that there is only one eigenvalue with norm 1, whose eigen-operator is the identity operator, as expected. All the other eigenvalues, though complex, have norm less than 1. This is manifestly seen in Figure \ref{fig:perfect tensor eigenvalues}. Since for the tensor product of two matrices, the eigenvalues are pairwise products of each matrix, we conclude that for the double copy of $2$-isometries, all but one eigenvalue have norm less than $1$ as well.

\begin{figure}[h]
    \centering
    \includegraphics[width=0.65\textwidth]{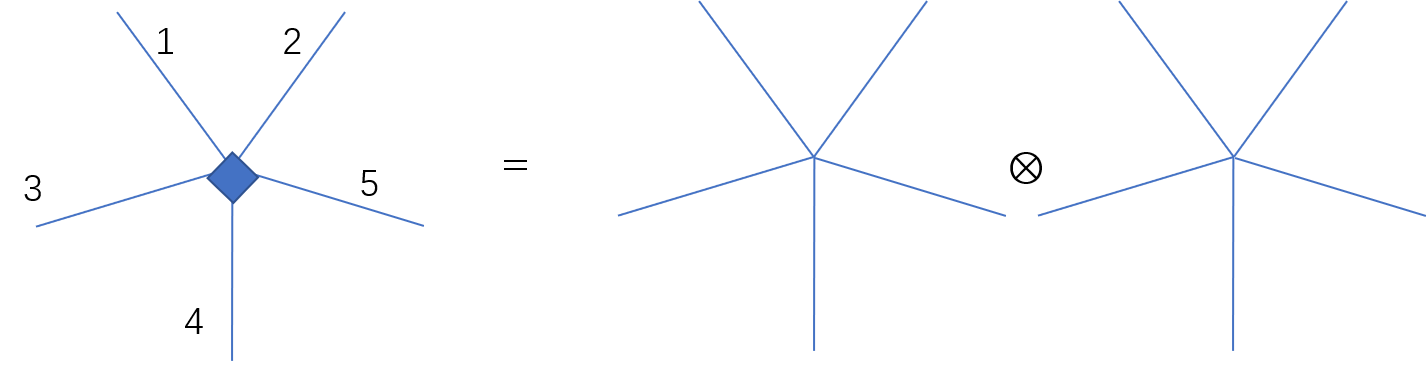}
    \includegraphics[width=0.34\textwidth]{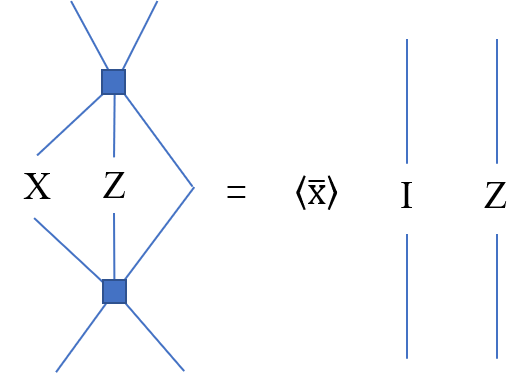}
    \caption{Left: The perfect tensor on the pentagon defines a 2-isometry. Right: An example of how the super-operator defined by $W_p (\cdot) W_p^\dagger$ acting on a weight 2 operator. The prefactor $\langle\bar{X}\rangle$ is the expectation value of the logical $\bar{X}$ operator in state $\ket{\phi}$.}
    \label{fig:perfect tensor}
\end{figure}
\begin{figure}
    \centering
    \includegraphics[width=0.5\textwidth]{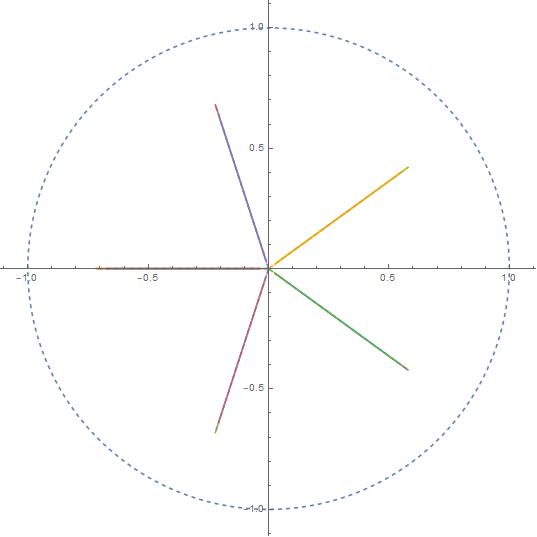}
    \caption{Eigenvalues of the super-operator defined by $W_p (\cdot) W_p^\dagger$. This is depicted on the complex plane with parameter $\alpha=\frac{\pi}{3},~\beta\in(0,2\pi)$. The outer dashed one is the unit circle. This manifestly shows that all but one eigenvalues have norm less than 1.}
    \label{fig:perfect tensor eigenvalues}
\end{figure}
 We then investigate the properties of the imperfect code defined in Figure \ref{fig:imperfect tensor illustration}. This defines a $1$-isometry. We denote the isometry as $W_I$. In this case, the original five-fold rotational symmetry breaks down to the reflection symmetry $2\leftrightarrow5,~ 3\leftrightarrow4$. It turns out that one can have weight 1 or 2 operator fed into the super-operator defined by $W_I(\cdot)W_I^\dagger$. In the weight 1 case, one can solve for its eigenvalues and eigen-operators. In the weight 2 case, operators can have support on leg $2,3$ or $3,4$ or $4,5$. In this case no eigen-operator can be defined because the operator weight changes. Instead, we calculate the operator norm after applying the super-operator $W_I(\cdot)W_I^\dagger$. The identity operator together with Pauli matrices of one qubit forms the bases of operators. For weight 2 operators, this is a Hilbert space of 256 dimensions. Since the bases operators all have norm 1, to show that the super-operator $W_I(\cdot)W_I^\dagger$ result in decaying two point functions, we  show the resulting operators have norm less than 1, except for the identity operators on 4 qubits. For operators inserted on leg $2,3$, numerical results in Figure \ref{fig:imperfect tensor norm} explicitly verifies this property. Inserting the operators in legs $3,4$ turns out to produce the same figure. Inserting the operators in legs $4,5$ is the same as inserting in $2,3$ by symmetry.
\begin{figure}
    \centering
    \includegraphics[width=0.7\textwidth]{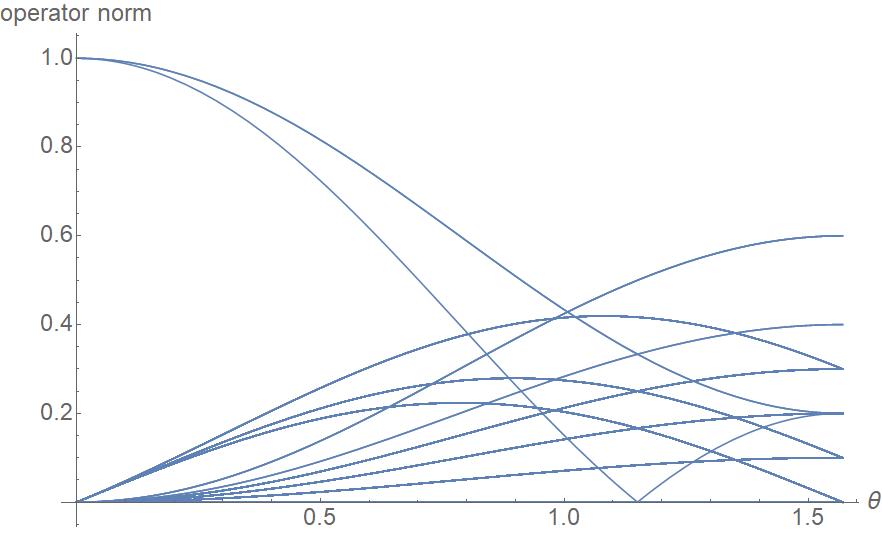}
    \caption{ The norms of weight 2 operators inserted on leg $2,3$ with respect to the parameter $\theta$ in the range $\theta\in (0,\frac{\pi}{2})$. We plot all 255 bases after the action of super-operator $W_I(\cdot)W_I^\dagger$. This explicitly shows that they have norm strictly less than 1 when $\theta>0$.}
    \label{fig:imperfect tensor norm}
\end{figure}

With the properties of perfect and imperfect codes above, we can already conclude that our model can produce decaying yet non-vanishing two point correlations. Furthermore, we would like to identify different types of super-operators. They are collectively depicted in Figure \ref{fig:superoperators}. We draw the dual graph so that the isometries are on the nodes. We summarize them in the following. A weight 1 operator can be coarse-grained either into a weight 1 operator via an imperfect code (panel 1), or into a weight 2 operator via a perfect tensor (panel 2). A weight 2 operator can be coarse-grained into a weight 1 operator via an imperfect tensor (panel 3), or into a weight 2 operator via a product of imperfect tensors (panel 4), or into a weight 3 tensor via a product of perfect and imperfect tensor (panel 5). A weight 3 operator may remain ``stable'' as weight 3 (panel 6 and 7), or ``shrink'' to weight 2 (panel 8). 

\begin{figure}
    \centering
    \includegraphics[width=0.75\textwidth]{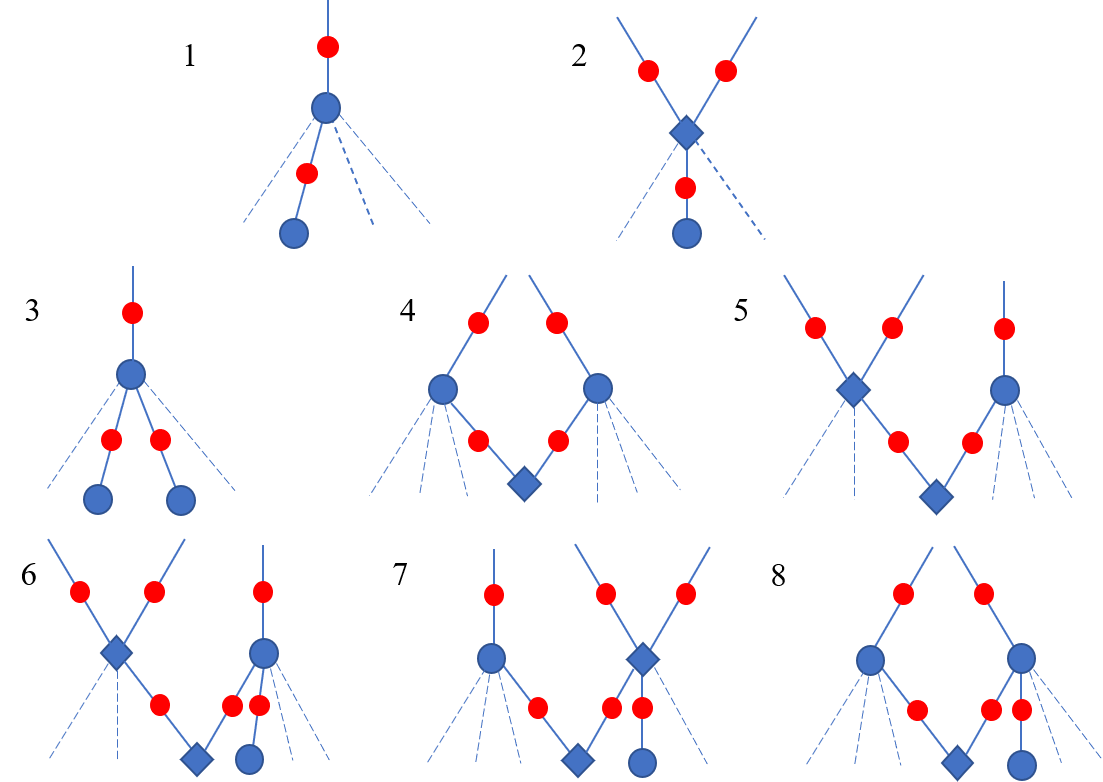}
    \caption{This shows the different superoperators during the coarse-graining process. The coarse-graining goes from bottom to top in each panel. Blue circles and diamonds represent imperfect and perfect tensors respectively. The red circles mark the legs on which the operators act. The irrelevant legs are depicted in dashed lines. The first row shows the cases of weight 1 operators, the second row for weight 2 and the third for weight 3.} 
    \label{fig:superoperators}
\end{figure}

Since the operator weight may change, it is not feasible to solve for the eigen-operators of super operators defined on an individual layer. It is possible to calculate the eigenvalue and the eigen-operators for the procedure during which a weight 1 operator increases to weight 3 then decreases to weight 1, i.e., defining a super-operator across multiple layers. However the number of layers involved depends largely on particular processes. We adopt a different strategy to calculate the probability of each type of super-operator appearing in a typical coarse graining process. For a given insertion point, the total coarse-graining process is fixed. We randomize the insertion point at the boundary for a generic operator so we can calculate these probabilities as an average over different coarse-graining processes\footnote{This is similar to the strategy in \cite{MERAalgo}, where the authors averaged over the super-operators, which can also be viewed as an averaging over different coarse-graining processes.}. An assumption is that the two operators inserted at the cut-off boundary are well-separated, i.e., $N\sim \log|i-j|\gg 1$, so the two insertions will be coarse-grained separately before they meet at the top few layers. It is this part that contributes to the long range behaviour of the correlation. Also, the notion of taking the average super-operator with said probability distribution make more sense when $N$ is large.

The types of super-operators and their corresponding probabilities depend primarily on the graph structure of the tensor network. 
We start by reviewing the formulae for the number of perfect tensors $g(n)$ and imperfect tensors $f(n)$ in each layer. One node is chosen as the center and is labelled layer $0$. Then if we need to go across a minimal number of $n$ edges to reach from the center to a particular node, it is assigned to layer $n$. The labelling procedure is unique and unambiguous in this model. This setup is familiar in the literature \cite{Pastawski:2015qua} and we reproduce its result here for readability:
\begin{equation}
\begin{split}
    f(n)&=\frac{5-\sqrt{5}}{2}\Big(\frac{3+\sqrt{5}}{2}\Big)^n\Big[1+O\Big(\Big(\frac{3-\sqrt{5}}{3+\sqrt{5}}\Big)^n\Big)\Big],\\
    g(n)&=\frac{3\sqrt{5}-5}{2}\Big(\frac{3+\sqrt{5}}{2}\Big)^n\Big[1+O\Big(\Big(\frac{3-\sqrt{5}}{3+\sqrt{5}}\Big)^n\Big)\Big].\\
\end{split}
\end{equation}
We work in the case where $n$ is large so we can set the factors in the square brackets to $1$. The number of total layers $N$ can be calculated more precisely as
\begin{equation}
    N\approx \frac{\log| i-j|}{\log(\frac{3+\sqrt{5}}{2})}.
\end{equation}
%\textcolor{blue}{YW: N is the total number of the layers and n is the labelling of the $nth$ layer.}

Now we are able to calculate the probability of a weight 1 operator expanding to weight 2. There are two situations. If the weight 1 operator is at the start of the coarse-graining process, the probability of it spreading is the same as the probability of inserting an operator on legs that  are directly connected to a perfect code.
\begin{equation}
    P_0(1\to 2)=\frac{3g(n)}{4f(n)+3g(n)}=3-\frac{6}{\sqrt{5}},~ P_0(1\to 1)=\frac{4f(n)}{4f(n)+3g(n)}=\frac{6}{\sqrt{5}}-2.
\end{equation}
This only affects the initial step and does not impact the final result in Eq. \ref{eqn:probability super operator}. If the weight 1 operator appears during  coarse-graining process, then it has a different probability. The reason is that the weight 1 operator necessarily came from the superoperator of an imperfect code. See Figure \ref{fig:weight 1} for an illustration. So this is the conditional probability of connecting to a perfect code at layer $n$, given that it is connected to imperfect code at layer $n+1$. 
\begin{equation}\label{eqn:probability weight 1}
     P(1\to 2|1)=\frac{g(n)}{2f(n)+g(n)}=\sqrt{5}-2,~
     P(1\to1|1)=\frac{2f(n)}{2f(n)+g(n)}=3-\sqrt{5}.
\end{equation}

\begin{figure}[h]
    \centering
    \includegraphics[width=0.8\textwidth]{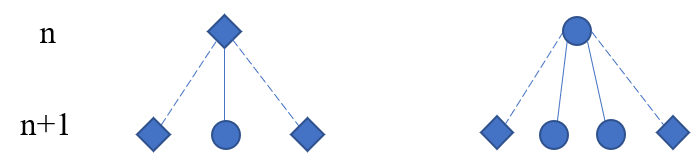}
    \caption{These figures illustrate the possible insertions of weight 1 operators. If the operator is inserted at the boundary, it can appear on any leg randomly. If it is a result of the coarse-graining, it can only appear on solid lines which connect to imperfect tensor in the layer $n+1$. It can not appear on dashed lines while being weight 1 in generic cases.}
    \label{fig:weight 1}
\end{figure}

If the operator is weight 2, then its behaviour is more complicated as it depends on the previous two coarse-graining steps.
Suppose an operator has just expanded from weight 1 to 2 by passing through a perfect code, it may either remain weight 2, shrink to weight 1, or expand to weight 3.  The illustrations are depicted in the second row of Figure \ref{fig:superoperators}. Ignoring the case where the operator weight remains constant, as it does not involve any transitions, from the graph one finds that when looking at layer $n$, there are $f(n-2)$ cases that corresponds to the panel $3$ of Figure \ref{fig:superoperators} and $2g(n-1)$ cases that corresponds to the panel $5$ of Figure \ref{fig:superoperators}. So this leads to the conditional probabilities: %\textcolor{red}{I find the presentation here unclear. The word ``number'' means the number of imperfect/perfect tensors at layer $n$ I suppose? But the ``former/latter'' cases are referring to transitions. On a unrelated note, I changed all imperfect/perfect tensors to imperfect/perfect codes to distinguish them from the 6 leg tensor forms. Please fix other instances if you find any.}
%\textcolor{blue}{YW:Slightly changed the phrases here}
\begin{equation}\label{eqn:probability weight 2}
    P(2\to2|1\to2)=\frac{f(n-2)}{g(n)}=\sqrt{5}-2,~ P(2\to3|1\to2)=\frac{2g(n-1)}{g(n)}=3-\sqrt{5}.
\end{equation}
If an operator has just shrunk from weight 3 or stayed as weight 2, it will necessarily shrink to weight 1. This is illustrated in panel 3, 4, 8 in Figure: \ref{fig:superoperators}.

\begin{equation}
    P(2\to1|2\to2)=P(2\to1|3\to2)=1.
\end{equation}

For a weight 3 operator, it is similar to the behaviour of weight 1 operator. It may stay weight 3 or shrink to weight 2 then to weight 1. This depends on which tensor these legs connect to in the next layer. Similar to weight 2 cases, if they both connect to imperfect tensors, the weight would shrink to 2; if they connect to 1 perfect and 1 imperfect tensor, the weight would remain to be 3. The illustrations are collected in the third row of Figure \ref{fig:superoperators}.
So we have the following conditional probabilities:
\begin{equation}\label{eqn:probability weight 3}
    P(3\to2|3)=\frac{f(n-2)}{g(n)}=\sqrt{5}-2,~ P(3\to3|3)=\frac{2g(n-1)}{g(n)}=3-\sqrt{5}.
\end{equation}

With all these conditional probabilities, we are able to calculate the unconditional probabilities of different types of super-operators appearing in a sufficiently long and typical coarse-graining process. We take $p(1\to1)=x,~p(3\to3)=y$. Then $p(1\to2)=\frac{p(1\to2|1)}{p(1\to1|1)}x$ from the conditional probabilities in Eq. \ref{eqn:probability weight 1}. Similarly $p(3\to2)=\frac{p(3\to2|3)}{p(3\to3|3)}y$ from Eq. \ref{eqn:probability weight 3}. Next we have $p(2\to1)=p(1\to2)$ and $p(2\to3)=p(3\to2)$. This is because when the total number of the layers $N$ is large enough and intermediate coarse-graining processes dominate. If an operator remains weight $1$ through several layers, it always begins with a transition $2\to 1$ and ends with a transition $1\to 2$. The same happens for the weight-$3$ case. The remaining unknown $p(2\to2)$ can be obtained from the normalization condition and from the conditional probability in Eq. \ref{eqn:probability weight 2}, such that $p(2\to2)=\frac{p(2\to2|1\to2)}{p(2\to3|1\to2)}p(2\to3)$. This yields two equations of $x,~y$ and we can solve for them to obtain Eq. \ref{eqn:probability super operator}.

The resulting probabilities are,
\begin{equation}\label{eqn:probability super operator}
    \begin{split}
        p(1\to1)&=\frac{3\sqrt{5}}{5}-1\approx34.16\%,~ p(1\to2)=1-\frac{2\sqrt{5}}{5}\approx 10.56\%,\\
        p(2\to1)&=1-\frac{2\sqrt{5}}{5}\approx 10.56\%,~p(2\to2)=\frac{9\sqrt{5}}{5}-4\approx2.49\%,~p(2\to3)=5-\frac{11\sqrt{5}}{5}\approx 8.07\%,\\
        p(3\to2)&=5-\frac{11\sqrt{5}}{5}\approx 8.07\%,~p(3\to3)=\frac{14\sqrt{5}}{5}-6\approx 26.10\%.
    \end{split}
\end{equation}
When $N$ is sufficiently large, there are approximately $p(i\to j)N$ layers in which the coarse-graining operation is sending a weight-$i$ operator to weight-$j$. These are precisely the probabilities $p_a$'s in Eq. \ref{eqn:averageeigenvalue}.

In summary, we have shown that each super-operator has a dominant singular value 1 with only the identity operator as its sole eigen-operator; all other singular values are less than 1. By averaging such super-operators, our HMERA model produces power-law decaying correlation functions. The probabilities of each type of super operator are further calculated in Eq. \ref{eqn:probability super operator}. This allows us to estimate the average eigenvalue $\bar{\lambda}$ that appears in in Eq. \ref{eqn:correlationinTN} and  \ref{eqn:averageeigenvalue}.

\subsection{Non-flat entanglement spectrum}
\label{subsec:entspec}
%\textcolor{red}{Yixu: This section may be moved to appendices}

It was pointed out in \cite{Dong:2018seb,Akers:2018fow} that, to leading order, the density matrix obtained from the gravitational path integral evaluated with fixed area RT surface has a flat entanglement spectrum. This implies that the Renyi entropies of the resulting density matrix is identical to all orders. It is consistent with the existing QECC constructions \cite{Pastawski:2015qua,Hayden:2016cfa}, yet in contradiction to any real CFT models \cite{cftspec}. In this section, we show that the tensor sub-networks in our proposed approximate QECC model can obtain non-flat entanglement spectra. %That is, the density matrices obtained from tensor contraction is not maximally mixed. Our model also provides a potential  where tunable parameters are available to fit real CFT entanglement spectra. 

Here we consider two examples. In Figure  \ref{fig:density matrix 3 type 1} and \ref{fig:density matrix 3 type 2} we calculate explicitly the eigenvalues of some reduced density matrices of three sites at the boundary. As each leg represents a qudit or a pair of qubits, the density matrices are of dimension 64. With two legs contracted, they both have rank 16. If the $\theta$ parameter is $0$, the eigenvalues are indeed flat, i.e. all the non-zero ones are $1/16$. In the figures we turn on $\theta$ to a generic random non-zero value and the spectrum is no longer flat.
 While they are insufficient to show that the entanglement spectrum is non-flat for all subsystems in the actual HMERA, we can gain some intuition from such toy examples.

\begin{figure}[ht]
    \centering
    \includegraphics[width=0.7\textwidth]{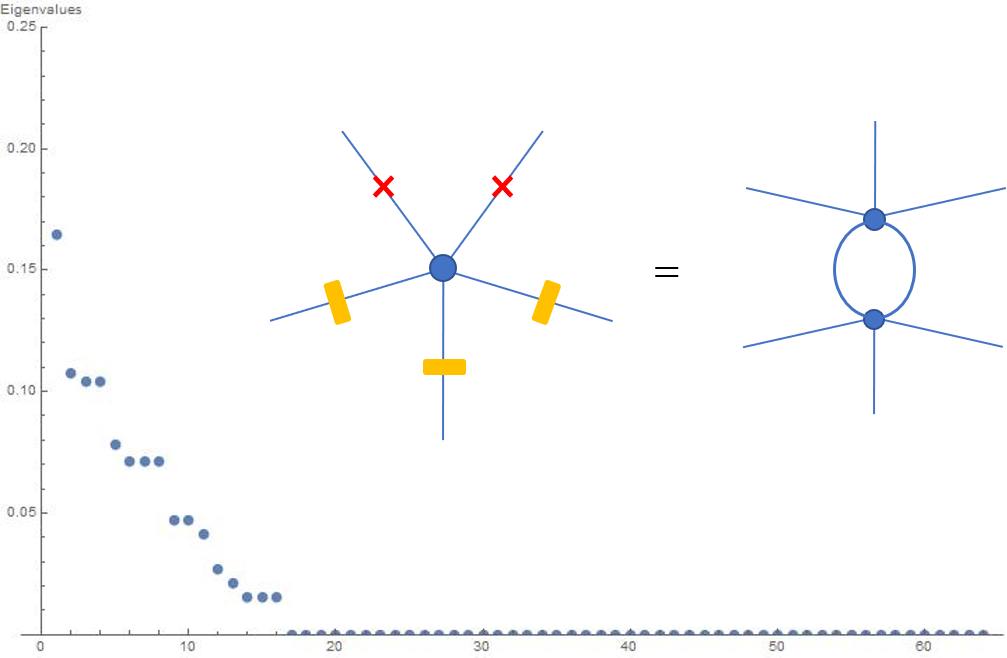}
    \caption{The reduced density matrix of the sites labelled by yellow rectangle are calculated. Its corresponding Schimdt coefficients are plotted on the left.  The legs that are labelled by red crosses are contracted to compute the reduced density matrix in the tensor network.} 
    \label{fig:density matrix 3 type 1}
\end{figure}

\begin{figure}[ht]
    \centering
    \includegraphics[width=0.7\textwidth]{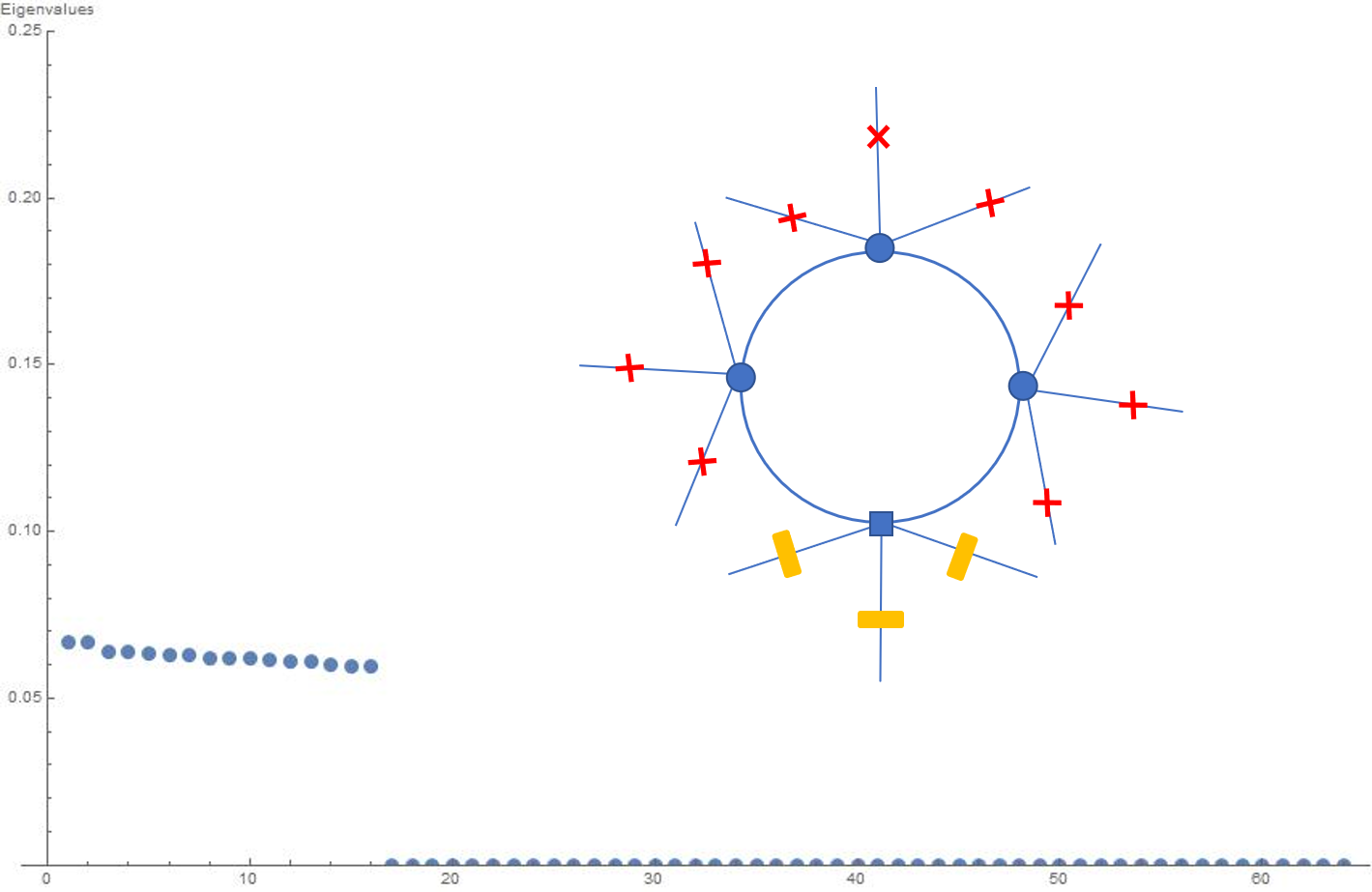}
    \caption{The eigenvalues of another reduced density matrix. The legends are the same as those in Figure. \ref{fig:density matrix 3 type 1}. } 
    \label{fig:density matrix 3 type 2}
\end{figure}

\subsection{Imperfect Tensor Properties}
\label{app:imptensor}
The imperfect tensor itself is an approximate quantum error correction code that inherits the code properties of two copies of the $[[5,1,3]]$ code in the small $\theta$ regime. For any value of $\theta\ne 0$, it is also an exact $1$-isometry when the logical qubits are fixed to be the $|\overline{00}\rangle$ state. While its isometric properties can be verified by showing all single-qudit (single-leg) reduced density matrices are maximally mixed, they can also be verified through tensor contraction. The latter technique also generalizes to other stabilizer codes without matrix computations. 

\begin{proof}
\begin{figure}[H]
    \centering
    \includegraphics[width=0.5\textwidth]{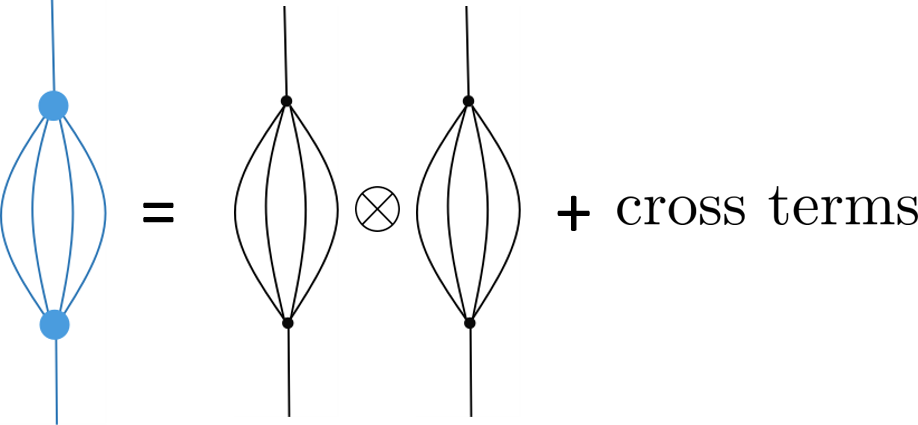}
    \caption{Contraction of two imperfect tensors (blue, left) with their logical states fixed in $|\bar{0}\rangle$ can be decomposed into contractions of perfect tensors with non-trivial operator insertions in the cross terms.}
    \label{fig:imp_contract}
\end{figure}
The imperfect tensor can be expanded as the sum over diagonal terms and cross terms~(Figure~\ref{fig:imp_contract}). Because we have chosen the coefficients to be normalized, the diagonal terms sum to the identity operator. It remains to show that the cross terms vanish. 

\begin{figure}[H]
    \centering
    \includegraphics[width=0.7\textwidth]{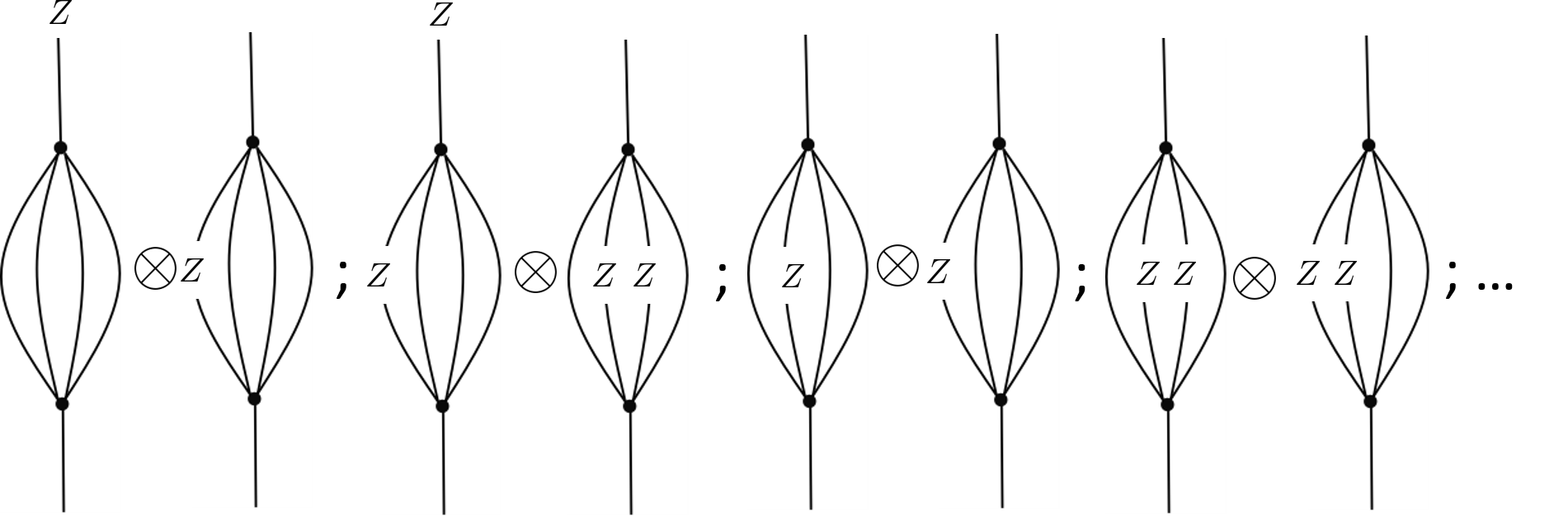}
    \caption{Examples of cross terms that are present in the contraction. The four diagrams shown here represent the four types of terms that are relevant to us. All other terms are variations of these terms with $Z$ operators in different places.}
    \label{fig:cross_terms}
\end{figure}

The cross terms consist of the type of tensor contractions shown in Figure~\ref{fig:cross_terms}.
All of them contain factors of the following two types: ones with a single $Z$ insertion (type-1) and ones with weight-2 $ZZ$ insertions (type-2).

The stabilizer group of the double copy perfect code is $\langle S_i\otimes I, I\otimes S_j\rangle$, where $S_i, i\in\{1,\dots,4\}$ are the original 5 qubit code stabilizer generators.
Suppose we can find a stabilizer element, $S=S'\otimes I$, that only acts non-trivially on the 4 contracted legs and anti-commute with the inserted operator $O$, then by the commutation relation in Figure~\ref{fig:anti-comstab}, the terms in Figure~\ref{fig:cross_terms} will vanish. 

\begin{figure}
    \centering
    \includegraphics[width=0.7\textwidth]{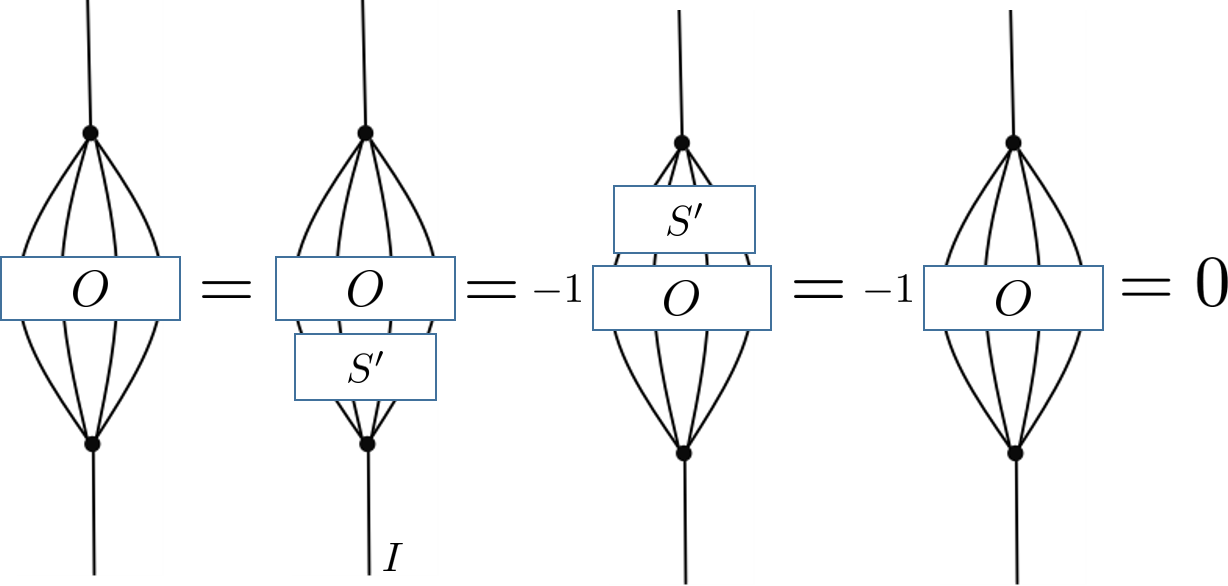}
    \caption{If $S=S'\otimes I$ and $\{S',O\}=0$, then the tensor contraction with operator $O$ inserted is trivial.}
    \label{fig:anti-comstab}
\end{figure}

Indeed, this can be trivially done for terms of type-1 --- let us treat the uncontracted leg(s) and the one potential Pauli Z error insertion within the contraction as two located errors.  Because the 5 qubit code detects two errors, there must exist stabilizers of the 5 qubit code that anticommute with any insertion of $P_i\otimes P_j$ where $P\in \{I, X, Y,Z\}$. Therefore, all terms in Figure~\ref{fig:cross_terms} vanish except for the last one where both of them are type-2 factors. 

For type-2 contraction with $ZZ$ insertions, the above argument no longer works because the code doesn't detect any three errors. To show that it vanishes, we do an exhaustive search of all stabilizers that will anti-commute with the terms with $ZZ$ insertions.

The stabilizer group of the 5 qubit code is 
\begin{equation}
    \mathcal{S} = \langle XZZXI, IXZZX, XIXZZ, ZXIXZ\rangle.
\end{equation}
Without loss of generality, let the first qubit be the uncontracted leg of the perfect tensor, then the only stabilizer elements from the 5 qubit code that act trivially on the uncontracted legs are $IXZZX, IYXXY,IZYYZ$ where we ignore the potential minus signs. These anti-commute with all weight-2 Z insertions except $IZIIZ, IIZZI$. However, recall that the state $|\bar{0}\rangle$ is also stabilized by $\bar{Z}=ZZZZZ$. A representation of $\bar{Z}$ is $IIYZY$, which anti-commutes with both $IZIIZ$ and $IIZZI$. Therefore all cross-terms vanish.

Note that other logical states are not stabilized by both $\bar{Z}\bar{I}$ and $\bar{I}\bar{Z}$. As there are no other stabilizers that anti-commute with $IZIIZ,IIZZI$, the tensor is not a $1$-isometry when it is in other logical states.

\end{proof}

Because the imperfect code in the $|\overline{00}\rangle$ state is a 1-isometry, any two qubits associated with one leg must be maximally entangled with the rest of the system. Because any leg is maximally mixed, it contains no information of the logical state. Thus, the code corrects any one qudit erasure error; it is a $[[5,0,2]]_4$ code on dimension-4 qudits. 

\end{appendix}

\bibliographystyle{unsrt}
\bibliography{paper}
\end{document}